\documentclass[12pt,a4paper,twoside,openright]{report}
\pdfoutput = 1
\usepackage{amsmath,amsthm,fancyhdr,fullpage,graphicx,setspace}
\newcommand{\Oh}[1]
    {\ensuremath{\mathcal{O} \hspace{-0.5ex} \left( {#1} \right)}}
\newcommand{\occ}[2]
    {\ensuremath{\mathrm{occ} \left( {#1}, {#2} \right)}}

\newtheorem{theorem}{Theorem}[chapter]
\newtheorem{lemma}[theorem]{Lemma}
\newtheorem{corollary}[theorem]{Corollary}
\newtheorem{problem}[theorem]{Open Problem}

\begin{document}

\pagestyle{empty}

\title{New Algorithms and Lower Bounds for\\Sequential-Access Data Compression}
\author{Travis Gagie}
\maketitle

\cleardoublepage

\begin{abstract}
This thesis concerns sequential-access data compression, i.e., by algorithms that read the input one or more times from beginning to end.  In one chapter we consider adaptive prefix coding, for which we must read the input character by character, outputting each character's self-delimiting codeword before reading the next one.  We show how to encode and decode each character in constant worst-case time while producing an encoding whose length is worst-case optimal.  In another chapter we consider one-pass compression with memory bounded in terms of the alphabet size and context length, and prove a nearly tight tradeoff between the amount of memory we can use and the quality of the compression we can achieve.  In a third chapter we consider compression in the read/write streams model, which allows us passes and memory both polylogarithmic in the size of the input.  We first show how to achieve universal compression using only one pass over one stream.  We then show that one stream is not sufficient for achieving good grammar-based compression.  Finally, we show that two streams are necessary and sufficient for achieving entropy-only bounds.
\end{abstract}

\setcounter{tocdepth}{1}
\tableofcontents
\thispagestyle{empty}
\cleardoublepage
\pagestyle{fancy}
\setcounter{page}{1}

\chapter{Introduction} \label{chp:introduction}

Sequential-access data compression is by no means a new subject, but it remains interesting both for its own sake and for the insight it provides into other problems.  Apart from the Data Compression Conference, several conferences often have tracks for compression (e.g., the International Symposium on Information Theory, the Symposium on Combinatorial Pattern Matching and the Symposium on String Processing and Information Retrieval), and papers on compression often appear at conferences on algorithms in general (e.g., the Symposium on Discrete Algorithms and the European Symposium on Algorithms) or even theory in general (e.g., the Symposium on Foundations of Computer Science, the Symposium on Theory of Computing and the International Colloquium on Algorithms, Languages and Programming).  We mention these conference in particular because, in this thesis, we concentrate on the theoretical aspects of data compression, leaving practical considerations for later.

Apart from its direct applications, work on compression has inspired the design and analysis of algorithms and data structures, e.g., succinct or compact data structures such as indexes.  Work on sequential data compression in particular has inspired the design and analysis of online algorithms and dynamic data structures, e.g., prediction algorithms for paging, web caching and computational finance.  Giancarlo, Scaturro and Utro~\cite{GSU} recently described how, in bioinformatics, compression algorithms are important not only for storage, but also for indexing, speeding up some dynamic programs, entropy estimation, segmentation, and pattern discovery.

In this thesis we study three kinds of sequential-access data compression: adaptive prefix coding, one-pass compression with memory bounded in terms of the alphabet size and context length, and compression in the read/write streams model.  Adaptive prefix coding is perhaps the most natural form of online compression, and adaptive prefix coders are the oldest and simplest kind of sequential compressors, having been studied for more than thirty years.  Nevertheless, in Chapter~\ref{chp:adaptive} we present the first one that is worst-case optimal with respect to both the length of the encoding it produces and the time it takes to encode and decode.  In Chapter~\ref{chp:comparisons} we observe that adaptive alphabetic prefix coding is equivalent to online sorting with binary comparisons, so our algorithm from Chapter~\ref{chp:adaptive} can easily be turned into an algorithm for online sorting.  Chapter~\ref{chp:sublinear} is also about online sorting but, instead of aiming to minimize the number of comparisons (which remains within a constant factor of optimal), we concentrate on trying to use sublinear memory, in line with research on streaming algorithms.  We then study compression with memory constraints because, although compression is most important when space is in short supply and compression algorithms are often implemented in limited memory, most analyses ignore memory constraints as an implementation detail, creating a gap between theory and practice.  We first study compression in the case when we can make only one pass over the data.  One-pass compressors that use memory bounded in terms of the alphabet size and context length can be viewed as finite-state machines, and in Chapter~\ref{chp:one-pass} we use that property to prove a nearly tight tradeoff between the amount of memory we can use and the quality of the compression we can achieve.  We then study compression in the read/write streams model, which allows us to make multiple passes over the data, change them, and even use multiple streams (see~\cite{Sch07}).  Streaming algorithms have revolutionized the processing of massive data sets, and the read/write streams model is an elegant conversion of the streaming model into a model of external memory.  By viewing read/write stream algorithms as simply more powerful automata, that can use passes and memory both polylogarithmic in the size of the input, in Chapter~\ref{chp:streaming} we prove lower bounds on the compression we can achieve with only one stream.  Specifically, we show that, although we can achieve universal compression with only one pass over one stream, we need at least two streams to achieve good grammar-based compression or entropy-only bounds.  We also combine previously known results to prove that two streams are sufficient for us to compute the Burrows-Wheeler Transform and, thus, achieve low-entropy bounds.  As corollaries of our lower bounds for compression, we obtain lower bounds for computing strings' minimum periods and for computing the Burrows-Wheeler Transform~\cite{BW94}, which came as something of a surprise to us.  It seems no one has previously considered the problem of finding strings' minimum periods in a streaming model, even though some related problems have been studied (see, e.g.,~\cite{EMS04}, in which the authors loosen the definition of a repeated substring to allow approximate matches).  We are currently investigating whether we can derive any more such results.

The chapters in this thesis were written separately and can be read separately; in fact, it might be better to read them with at least a small pause between chapters, so the variations in the models considered do not become confusing.  Of course, to make each chapter independent, we have had to introduce some degree of redundancy.  Chapter~\ref{chp:adaptive} was written specifically for this thesis, and is based on recent joint work~\cite{GN} with Yakov Nekrich at the University of Bonn; a summary was presented at the University of Bielefeld in October of 2008, and will be presented at the annual meeting of the Italy-Israel FIRB project ``Pattern Discovery in Discrete Structures, with Applications to Bioinformatics'' at the University of Palermo in February of 2009.  Chapter~\ref{chp:comparisons} was also written specifically for this thesis, but Chapter~\ref{chp:sublinear} was presented at the 10th Italian Conference on Theoretical Computer Science~\cite{Gag07b} and then published in {\em Information Processing Letters}~\cite{Gag08}. Chapter~\ref{chp:one-pass} is a slight modification of part of a paper~\cite{GM07b} written with Giovanni Manzini at the University of Eastern Piedmont, which was presented in 2007 at the 32nd Symposium on Mathematical Foundations of Computer Science.  A paper~\cite{GKN09} we wrote with Marek Karpinski (also at the University of Bonn) and Yakov Nekrich that partially combines the results in these two chapters, will appear at the 2009 Data Compression Conference.  Chapter~\ref{chp:streaming} is a slight modification of a paper~\cite{Gagb} that has been submitted to a conference, with a very brief summary of some material from a paper~\cite{GM07a} written with Giovanni Manzini and presented at the 18th Symposium on Combinatorial Pattern Matching.

\chapter{Adaptive Prefix Coding} \label{chp:adaptive}

Prefix codes are sometimes called instantaneous codes because, since no codeword is a prefix of another, the decoder can output each character once it reaches the end of its codeword.  Adaptive prefix coding could thus be called ``doubly instantaneous'', because the encoder must produce a self-delimiting codeword for each character before reading the next one.  The main idea behind adaptive prefix coding is simple: both the encoder and the decoder start with the same prefix code; the encoder reads the first character of the input and writes its codeword; the decoder reads that codeword and decodes the first character; the encoder and decoder now have the same information, and they update their copies of the code in the same way; then they recurse on the remainder of the input.  The two main challenges are, first, to efficiently update the two copies of the code and, second, to prove the total length of the encoding is not much more than what it would be if we were to use an optimal static prefix coder.

Because Huffman codes~\cite{Huf52} have minimum expected codeword length, early work on adaptive prefix coding naturally focused on efficiently maintaining a Huffman code for the prefix of the input already encoded.  Faller~\cite{Fal73}, Gallager~\cite{Gal78} and Knuth~\cite{Knu85} developed an adaptive Huffman coding algorithm --- usually known as the FGK algorithm, for their initials --- and showed it takes time proportional to the length of the encoding it produces.  Vitter~\cite{Vit87} gave an improved version of their algorithm and showed it uses less than one more bit per character than we would use with static Huffman coding.  (For simplicity we consider only binary encodings; therefore, by $\log$ we always mean $\log_2$.)  Milidi\'u, Laber and Pessoa~\cite{MLP99} later extended Vitter's techniques to analyze the FGK algorithm, and showed it uses less than two more bits per character than we would use with static Huffman coding.  Suppose the input is a string $s$ of $n$ characters drawn from an alphabet of size \(\sigma \ll n\); let $H$ be the empirical entropy of $s$ (i.e., the entropy of the normalized distribution of characters in $s$) and let $r$ be the redundancy of a Huffman code for $s$ (i.e., the difference between the expected codeword length and $H$).  The FGK algorithm encodes $s$ as at most \((H + 2 + r) n + o (n)\) bits and Vitter's algorithm encodes it as at most \((H + 1 + r) n + o (n)\) bits; both take $\Oh{(H + 1) n}$ total time to encode and decode, or $\Oh{H + 1}$ amortized time per character.  Table~\ref{tab:bounds} summarizes bounds for various adaptive prefix coders.

If $s$ is drawn from a memoryless source then, as $n$ grows, adaptive Huffman coding will almost certainly ``lock on'' to a Huffman code for the source and, thus, use \((H + r) n + o (n)\) bits.  In this case, however, the whole problem is easy: we can achieve the same bound, and use less time, by periodically building a new Huffman code.  If $s$ is chosen adversarially, then every algorithm uses at least \((H + 1 + r) n - o (n)\) bits in the worst case.  To see why, fix an algorithm and suppose \(\sigma = n^{1 / 2} = 2^\ell + 1\) for some integer $\ell$, so any binary tree on $\sigma$ leaves has at least two leaves with depths at least \(\ell + 1\).  Any prefix code for $\sigma$ characters can be represented as a code-tree on $\sigma$ leaves; the length of the lexicographically $i$th codeword is the depth of the $i$th leaf from the left.  It follows that an adversary can choose $s$ such that the algorithm encodes it as at least \((\ell + 1) n\) bits.  On the other hand, a static prefix coding algorithm can assign codewords of length $\ell$ to the \(\sigma - 2\) most frequent characters and codewords of length \(\ell + 1\) to the two least frequent characters, and thus use at most \(\ell n + 2 n / \sigma + \Oh{\sigma \log \sigma} = \ell n + o (n)\) bits.  Therefore, since the expected codeword length of a Huffman code is minimum, \((H + r) n \leq \ell n + o (n)\) and so \((\ell + 1) n \geq (H + 1 + r) n - o (n)\).

This lower bound seems to say that Vitter's upper bound cannot be significantly improved.  However, to force the algorithm to use \((H + 1 + r) n - o (n)\) bits, it might be that the adversary must choose $s$ such that $r$ is small.  In a previous paper~\cite{Gag07a} we were able to show this is the case, by giving an adaptive prefix coder that encodes $s$ in at most \((H + 1) n + o (n)\) bits.  This bound is perhaps a little surprising, since our algorithm was based on Shannon codes~\cite{Sha48}, which generally do not have minimum expected codeword length.  Like the FGK algorithm and Vitter's algorithm, our algorithm used $\Oh{H + 1}$ amortized time per character to encode and decode.  Recently, Karpinski and Nekrich~\cite{KN} showed how to combine some of our results with properties of canonical codes, defined by Schwartz and Kallick~\cite{SK64} (see also~\cite{Kle00,TM00,TM01}), to achieve essentially the same compression while encoding and decoding each character in $\Oh{1}$ amortized time and $\Oh{\log (H + 2)}$ amortized time, respectively.  Nekrich~\cite{Nek07} implemented their algorithm and observed that, in practice, it is significantly faster than arithmetic coding and slightly faster Turpin and Moffat's GEO coder~\cite{TM01}, although the compression it achieves is not quite as good.  The rest of this chapter is based on joint work with Nekrich~\cite{GN} that shows how, on a RAM with \(\Omega (\log n)\)-bit words, we can speed up our algorithm even more, to both encode and decode in $\Oh{1}$ worst-case time.  We note that Rueda and Oommen~\cite{RO04,RO06,RO08} have demonstrated the practicality of certain implementations of adaptive Fano coding, which is somewhat related to our algorithm ,especially a version~\cite{Gag04} in which we maintained an explicit code-tree.

\begin{table}[h]
\caption{Bounds for adaptive prefix coding: the times to encode and decode each character and the total length of the encoding.  Bounds in the first row and last column are worst-case; the others are amortized.}
\label{tab:bounds}
\resizebox{\textwidth}{!}
{\begin{tabular}{r|ccc}
\noalign{\bigskip}
& Encoding & Decoding & Length\\
\hline \rule{0ex}{3ex}
Gagie and Nekrich~\cite{GN} & $\Oh{1}$ & $\Oh{1}$ & \((H + 1) n + o (n)\)\\[1ex]
Karpinski and Nekrich~\cite{KN} & $\Oh{1}$ & $\Oh{\log (H + 2)}$ & \((H + 1) n + o (n)\)\\[1ex]
Gagie~\cite{Gag07a} & $\Oh{H + 1}$ & $\Oh{H + 1}$ & \((H + 1) n + o (n)\)\\[1ex]
Vitter~\cite{Vit87} & $\Oh{H + 1}$ & $\Oh{H + 1}$ & \((H + 1 + r) n + o (n)\)\\[1ex]
Knuth~\cite{Knu85} &&&\\[1ex]
Gallager~\cite{Gal78} & \raisebox{0ex}[0ex][0ex]{$\left. \rule{0ex}{4ex} \right\}$} $\Oh{H + 1}$ & $\Oh{H + 1}$ & \((H + 2 + r) n + o (n)\)\\[1ex]
Faller~\cite{Fal73} &&&
\end{tabular}}
\end{table}

\section{Algorithm} \label{sec:adapt alg}

A Shannon code is one in which, if a character has probability $p$, then its codeword has length at most \(\lceil \log (1 / p) \rceil\).  In his seminal paper on information theory, Shannon~\cite{Sha48} showed how to build such a code for any distribution containing only positive probabilities.

\begin{theorem}[Shannon, 1948] \label{thm:Sha48}
Given a probability distribution \(P = p_1, \ldots, p_\sigma\) with \(p_1 \geq \cdots \geq p_\sigma > 0\), we can build a prefix code in $\Oh{\sigma}$ time whose codewords, in lexicographic order, have lengths \(\lceil \log (1 / p_1) \rceil, \ldots, \lceil \log (1 / p_\sigma) \rceil\).
\end{theorem}

We encode each character of $s$ using a canonical Shannon code for a probability distribution that is roughly the normalized distribution of characters in the prefix of $s$ already encoded.  In order to avoid having to consider probabilities equal to 0, we start by assigning every character a count of 1.  This means the smallest probability we ever consider is at least \(1 / (n + \sigma)\), so the longest codeword we ever consider is $\Oh{\log n}$ bits.

A canonical code is one in which the first codeword is a string of 0s and, for \(1 \leq i < \sigma\), we can obtain the \((i + 1)\)st codeword by incrementing the $i$th codeword (viewed as a binary number) and appending some number of 0s to it.  For example, Figure~\ref{fig:canonical} shows the codewords in a canonical code, together with their lexicographic ranks.  By definition, the difference between two codewords of the same length in a canonical code, viewed as binary numbers, is the same as the difference between their ranks.  For example, the third codeword in Figure~\ref{fig:canonical} is 0100 and the sixth codeword is 0111, and \((0111)_2 - (0100)_2 = 6 - 3\), where $(\cdot)_2$ means the argument is to be viewed as a binary number.  We use this property to build a representation of the code that lets us quickly answer encoding and decoding queries.

\begin{figure}[h]
\[\begin{array}{rl@{\hspace{5ex}}rl}
1) & 000 & 7) & 1000\\
2) & 001 & 8) & 1001\\
3) & 0100 & 9) & 10100\\
4) & 0101 & 10) & 10101\\
5) & 0110 & \multicolumn{2}{c}{\raisebox{-.5ex}[0ex][0ex]{\vdots}}\\
6) & 0111 & 16) & 11011
\end{array}\]
\caption{The codewords in a canonical code.}
\label{fig:canonical}
\end{figure}

We maintain the following data structures: an array $A_1$ that stores the codewords' ranks in lexicographic order by character; an array $A_2$ that stores the characters and their frequencies in order by frequency; a dictionary $D_1$ that stores the rank of the first codeword of each length, with the codeword itself as auxiliary information; and a dictionary $D_2$ that stores the first codeword of each length, with its rank as auxiliary information.

To encode a given character $a$, we first use $A_1$ to find $a$'s codeword's rank; then use $D_1$ to find the rank of the first codeword of the same length and that codeword as auxiliary information; then add the difference in ranks to that codeword, viewed as a binary number, to obtain $a$'s codeword.  For example, if the codewords are as shown in Figure~\ref{fig:canonical} and $a$ is the $j$th character in the alphabet and has codeword 0111, then
\begin{enumerate}
\item \(A_1 [j] = 6\),
\item \(D_1.\mathrm{pred} (6) = \langle 3, 0100 \rangle\),
\item \((0100)_2 + 6 - 3 = (0111)_2\).
\end{enumerate}

To decode $a$ given a binary string prefixed with its codeword, we first search in $D_2$ for the predecessor of the first \(\lceil \log (n + \sigma) \rceil\) bits of the binary string, to find the first codeword of the same length and that codeword's rank as auxiliary information; then add that rank to the difference in codewords, viewed as binary numbers, to obtain $a$'s codeword's rank; then use $A_2$ to find $a$.  In our example above,
\begin{enumerate}
\item \(D_2.\mathrm{pred} (0111\ldots) = \langle 0100, 3 \rangle\),
\item \(3 + (0111)_2 - (0100)_2 = 6\),
\item \(A_2 [6] = a_j\).
\end{enumerate}
Admittedly, reading the first \(\lceil \log (n + \sigma) \rceil\) bits of the binary string will generally result in the decoder reading past the end of most codewords before outputting the corresponding character.  We do not see how to avoid this without potentially having the decoder read some codewords bit by bit.

Querying $A_1$ and $A_2$ takes $\Oh{1}$ worst-case time, so the time to encode and decode depends mainly on the time needed for predecessor queries on $D_1$ and $D_2$.  Since the longest codeword we ever consider is $\Oh{\log n}$ bits, each dictionary contains $\Oh{\log n}$ keys, so we can implement each as an instance of the data structure described below, due to Fredman and Willard~\cite{FW93}.  This way, apart from the time to update the dictionaries, we encode and decode each character in a total of $\Oh{1}$ worst-case time.  Andersson, Miltersen and Thorup~\cite{ABT99} showed Fredman and Willard's data structure can be implemented with AC$^0$ instructions, and Thorup~\cite{Tho03} showed it can be implemented with AC$^0$ instructions available on a Pentium 4; admittedly, though, in practice it might still be faster to use a sorted array to encode and decode each character in $\Oh{\log \log n}$ time.

\begin{lemma}[Fredman \& Willard, 1993] \label{lem:FW93}
Given $\Oh{\log^{1 / 6} n}$ keys, we can build a dictionary in $\Oh{\log^{2 / 3} n}$ worst-case time that stores those keys and supports predecessor queries in $\Oh{1}$ worst-case time.
\end{lemma}

\begin{corollary} \label{cor:FW93}
Given $\Oh{\log n}$ keys, we can build a dictionary in $\Oh{\log^{3 / 2} n}$ worst-case time that stores those keys and supports predecessor queries in $\Oh{1}$ worst-case time.
\end{corollary}

\begin{proof}
We store the keys at the leaves of a search tree with degree $\Oh{\log^{1 / 6} n}$, size $\Oh{\log^{5 / 6} n}$ and height at most 5.  Each node stores an instance of Fredman and Willard's dictionary from Lemma~\ref{lem:FW93}: each dictionary at a leaf stores $\Oh{\log^{1 / 6} n}$ keys and each dictionary at an internal node stores the first key in each of its children's dictionaries.  It is straightforward to build the search tree in \(\Oh{\log^{2 / 3 + 5 / 6} n} = \Oh{\log^{3 / 2} n}\) time and implement queries in $\Oh{1}$ time.
\end{proof}

Since a codeword's lexicographic rank is the same as the corresponding character's rank by frequency, and a character's frequency is an integer than changes only by being incremented after each of its occurrence, we can use a data structure due to Gallager~\cite{Gal78} to update $A_1$ and $A_2$ in $\Oh{1}$ worst-case time per character of $s$.  We can use $\Oh{\log n}$ binary searches in $A_2$ and $\Oh{\log^2 n}$ time to compute the number of codewords of each length; building $D_1$ and $D_2$ then takes $\Oh{\log^{3 / 2} n}$ time.  Using multiple copies of each data structure and standard background-processing techniques, we update each set of copies after every \(\lfloor \log^2 n \rfloor\) characters and stagger the updates, such that we need spend only $\Oh{1}$ worst-case time per character and, for \(1 \leq i \leq n\), the copies we use to encode the $i$th character of $s$ will always have been last updated after we encoded the \((i - \lfloor \log^2 n \rfloor)\)th character.

Writing \(s [i]\) for the $i$th character of $s$, \(s [1..i]\) for the prefix of $s$ of length $i$, and $\occ{s [i]}{s [1..i]}$ for the number of times \(s [i]\) occurs in \(s [1..i]\), we can summarize the results of this section as the following lemma.

\begin{lemma} \label{lem:adapt alg}
For \(1 \leq i \leq n\), we can encode \(s [i]\) as at most
\[\left\lceil \log \frac{i + \sigma}
    {\max \left( \occ{s [i]}{s [1..i]} - \lfloor \log^2 n \rfloor,
    1 \right)} \right\rceil\]
bits such that encoding and decoding it take $\Oh{1}$ worst-case time.
\end{lemma}

\section{Analysis}

Analyzing the length of the encoding our algorithm produces is just a matter of bounding the sum of the codewords' lengths.  Fortunately, we can do this using a modification of the proof that adaptive arithmetic coding produces an encoding not much longer than the one decrementing arithmetic coding produces (see, e.g.,~\cite{HV92}).

\begin{lemma} \label{lem:analysis}
\[\sum_{i = 1}^n \left\lceil \log \frac{i + \sigma}
    {\max \left( \occ{s [i]}{s [1..i]} - \lfloor \log^2 n \rfloor,
    1 \right)} \right\rceil
\leq (H + 1) n + \Oh{\sigma \log^3 n}\,.\]
\end{lemma}

\begin{proof}
Since \(\left\{ \rule{0ex}{2ex} \occ{s [i]}{s [1..i]}\,:\,1 \leq i \leq n \right\}\) and \(\left\{ j\,:\,\parbox{18ex}{\(1 \leq j \leq \occ{a}{s}\), \newline \mbox{$a$ a character}} \right\}\) are the same multiset,
\begin{equation*}
\begin{split}
& \sum_{i = 1}^n \left\lceil \log \frac{i + \sigma}
    {\max \left( \occ{s [i]}{s [1..i]} - \lfloor \log^2 n \rfloor, 1 \right)} \right\rceil\\
& < \sum_{i = 1}^n \log (i + \sigma) -
    \sum_{i = 1}^n \log \max \left( \occ{s [i]}{s [1..i]} - \lfloor \log^2 n \rfloor, 1 \right) + n\\
& = \sum_{i = 1}^n \log (i + \sigma) -
    \sum_a \sum_{j = 1}^{\occ{a}{s} - \lfloor \log^2 n \rfloor} \log j + n\\
& < \sum_{i = 1}^n \log i + \sigma \log (n + \sigma) -
    \sum_a \sum_{j = 1}^{\occ{a}{s}} \log j + \sigma \log^3 n + n\\
& = \log (n!) - \sum_a \log (\occ{a}{s}!) + n + \Oh{\sigma \log^3 n}\,.
\end{split}
\end{equation*}
Since
\[\log (n!) - \sum_a \log (\occ{a}{s}!)
= \log \frac{n!}{\prod_a \occ{a}{s}!}\]
is the number of distinct arrangements of the characters in $s$, we could complete the proof by information-theoretic arguments; however, we will use straightforward calculation.  Specifically, Robbins' extension~\cite{Rob55} of Stirling's Formula,
\[\sqrt{2 \pi} x^{x + 1 / 2} e^{-x + 1 / (12 x + 1)}
< x!
< \sqrt{2 \pi} x^{x + 1 / 2} e^{-x + 1 / (12 x)}\,,\]
implies that
\[x \log x - x \log e
< \log (x!)
\leq x \log x - x \log e + \Oh{\log x}\,,\]
where $e$ is the base of the natural logarithm.  Therefore, since \(\sum_a \occ{a}{s} = n\),
\begin{equation*}
\begin{split}
& \log (n!) - \sum_a \log (\occ{a}{s}!) + n + \Oh{\sigma \log^3 n}\\
& = n \log n - \sum_a \occ{a}{s} \log \occ{a}{s} -\\
& \quad n \log e + \sum_a \occ{a}{s} \log e + n + \Oh{\sigma \log^3 n}\\
& = \sum_a \occ{a}{s} \log \frac{n}{\occ{a}{s}} + n + \Oh{\sigma \log^3 n}\\
& = (H + 1) n + \Oh{\sigma \log^3 n}\,. \qedhere
\end{split}
\end{equation*}
\end{proof}

Combining Lemmas~\ref{lem:adapt alg} and~\ref{lem:analysis} and assuming \(\sigma = o (n / \log^3 n)\) immediately gives us our result for this chapter.

\begin{theorem} \label{thm:adaptive}
We can encode $s$ as at most \((H + 1) n + o (n)\) bits such that encoding and decoding each character takes $\Oh{1}$ worst-case time.
\end{theorem}

\chapter{Online Sorting with Few Comparisons} \label{chp:comparisons}

Comparison-based sorting is perhaps the most studied problem in computer science, but there remain basic open questions about it.  For example, exactly how many comparisons does it take to sort a multiset?  Over thirty years ago, Munro and Spira~\cite{MS76} proved distribution-sensitive upper and lower bounds that differ by $\Oh{n \log \log \sigma}$, where $n$ is the size of the multiset $s$ and $\sigma$ is the number of distinct elements $s$ contains.  Specifically, they proved that \(n H + \Oh{n}\) ternary comparisons are sufficient and \(n H - (n - \sigma) \log \log \sigma - \Oh{n}\) are necessary, where \(H = \sum_a \frac{\occ{a}{s}}{n} \log \frac{n}{\occ{a}{s}}\) denotes the entropy of the distribution of elements in $s$ and $\occ{a}{s}$ denotes the multiplicity of the distinct element $a$ in $s$. Throughout, by $\log$ we mean $\log_2$.  Their bounds have been improved in a series of papers, summarized in Table~\ref{tab:comp bounds}, so that the difference now is slightly less than \((1 + \log e) n \approx 2.44 n\), where $e$ is the base of the natural logarithm.

Apart from the bounds shown in Table~\ref{tab:comp bounds}, there have been many bounds proven about, e.g., sorting multisets in-place or with minimum data movement, or in the external-memory or cache-oblivious models.  In this chapter we consider online stable sorting; online algorithms sort $s$ element by element and keep those already seen in sorted order, and stable algorithms preserve the order of equal elements.  For example, splaysort (i.e., sorting by insertion into a splay tree~\cite{ST85}) is online, stable and takes $\Oh{(H + 1) n}$ comparisons and time.  In Section~\ref{sec:comp alg} we show how, if \(\sigma = o (n^{1 / 2} / \log n)\), then we can sort $s$ online and stably using \((H + 1) n + o (n)\) ternary comparisons and $\Oh{(H + 1) n}$ time.  In Section~\ref{sec:comp lbound} we prove \((H + 1) n - o (n)\) comparisons are necessary in the worst case.

\begin{table}[h]
\caption{Bounds for sorting a multiset using ternary comparisons.}
\label{tab:comp bounds}
\resizebox{\textwidth}{!}
{\begin{tabular}{r|cc}
\noalign{\bigskip}
& Upper bound & Lower bound\\
\hline \rule{0ex}{3ex}
Munro and Raman~\cite{MR91} && \((H - \log e) n + \Oh{\log n}\)\\[1ex]
Fischer~\cite{Fis84} & \((H + 1) n - \sigma\) & \((H - \log H) n - \Oh{n}\)\\[1ex]
Dobkin and Munro~\cite{DM80} && \(\left( H - n \log \left( \log n - \frac{\sum_a \occ{a}{s} \log \occ{a}{s}}{n} \right) \right) n - \Oh{n}\)\\[1ex]
Munro and Spira~\cite{MS76} & \(n H + \Oh{n}\) & \(n H - (n - \sigma) \log \log \sigma - \Oh{n}\)
\end{tabular}}
\end{table}

\section{Algorithm} \label{sec:comp alg}

Our idea is to sort $s$ by inserting its elements into a binary search tree $T$, which we rebuild occasionally using the following theorem by Mehlhorn~\cite{Meh77}.  We rebuild $T$ whenever the number of elements processed since the last rebuild is equal to the number of distinct elements seen by the time of the last rebuild.  This way, we spend $\Oh{n}$ total time rebuilding $T$.

\begin{theorem}[Mehlhorn, 1977] \label{thm:Meh77}
Given a probability distribution \(P = p_1, \ldots, p_k\) on $k$ keys, with no \(p_i = 0\), in $\Oh{k}$ time we can build a binary search tree containing those keys at depths at most \(\log (1 / p_1), \ldots, \log (1 / p_k)\).
\end{theorem}

To rebuild $T$ after processing $i$ elements of $s$, to each distinct element $a$ seen so far we assign probability \(\occ{a}{s [1..i]} / i\), where \(s [1..i]\) denotes the first $i$ elements of $s$; we then apply Theorem~\ref{thm:Meh77}.  Notice the smallest probability we consider is at least \(1 / n\), so the resulting tree has height at most \(\log n\).

We want $T$ always to contain a node for each distinct element $a$ seen so far, that stores $a$ as a key and a linked list of $a$'s occurrences so far.  After we use Mehlhorn's theorem, therefore, we extend $T$ and then replace each of its leaf by an empty AVL tree~\cite{AL62}.  To process an element \(s [i]\) of $s$, we search for \(s [i]\) in $T$; if we find a node $v$ whose key is equal to \(s [i]\), then we append \(s [i]\) to $v$'s linked list; otherwise, our search ends at a node of an AVL tree, into which we insert a new node whose key is equal to \(s [i]\) and whose linked list contains \(s [i]\).

If an element equal to \(s [i]\) occurs by the time of the last rebuild before we process \(s [i]\), then the corresponding node is at depth at most \(\log \frac{i}{\max \left( \rule{0ex}{2ex} \occ{s [i]}{s [1..i]} - \sigma, 1 \right)}\) in $T$, so the number of ternary comparisons we use to insert \(s [i]\) into $T$ is at most that number plus 1; the extra comparison is necessary to check that the algorithm should not proceed deeper into the tree.  Otherwise, since our AVL trees always contain at most $\sigma$ nodes, we use \(\Oh{\log n + \log \sigma} = \Oh{\log n}\) comparisons.  Therefore, we use a total of at most
\[\sum_{i = 1}^n \log \frac{i}{\max \left( \rule{0ex}{2ex} \occ{s [i]}{s [1..i]} - \sigma, 1 \right)} + n + \Oh{\sigma \log n}\]
comparisons to sort $s$ and, assuming each comparison takes $\Oh{1}$ time, a proportional amount of time.  We can bound this sum using the following technical lemma, which says the logarithm of the number of distinct arrangements of the elements in $s$ is close to \(n H\), and the subsequent corollary.  We write \(a_1, \ldots, a_\sigma\) to denote the distinct elements in $s$.

\begin{lemma} \label{lem:multinomial}
\[n H - \Oh{\sigma \log (n / \sigma)}
\leq \log {n! \choose \occ{a_1}{s}!, \ldots, \occ{a_\sigma}{s}!}
\leq n H + \Oh{\log n}\,.\]
\end{lemma}

\begin{proof}
Robbins' extension~\cite{Rob55} of Stirling's Formula,
\[\sqrt{2 \pi} x^{x + 1 / 2} e^{-x + 1 / (12 x + 1)}
< x!
< \sqrt{2 \pi} x^{x + 1 / 2} e^{-x + 1 / (12 x)}\,,\]
implies that
\[x \log x - x \log e
< \log (x!)
\leq x \log x - x \log e + \Oh{\log x}\,.\]
Therefore, since \(\sum_a \occ{a}{s} = n\), straightforward calculation shows that
\[\log {n! \choose \occ{a_1}{s}!, \ldots, \occ{a_\sigma}{s}!}
= \log (n!) - \sum_a \log (\occ{a}{s}!)\]
is at least \(n H - \Oh{\sigma \log (n / \sigma)}\) and at most \(\leq n H + \Oh{\log n}\).
\end{proof}

\begin{corollary} \label{cor:multinomial}
\begin{equation*}
\begin{split}
& \sum_{i = 1}^n \log \frac{i}{\max \left( \rule{0ex}{2ex} \occ{s [i]}{s [1..i]} - \sigma, 1 \right)} + n + \Oh{\sigma \log n}\\[1ex]
& \leq (H + 1) n + \Oh{\sigma^2 \log n}\,.
\end{split}
\end{equation*}
\end{corollary}

\begin{proof}
Since \(\left\{ \rule{0ex}{2ex} \occ{s [i]}{s [1..i]}\,:\,1 \leq i \leq n \right\}\) and \(\left\{ j\,:\,\parbox{18ex}{\(1 \leq j \leq \occ{a}{s}\), \newline \mbox{$a$ an element}} \right\}\) are the same multiset,
\begin{equation*}
\begin{split}
& \sum_{i = 1}^n \log \frac{i}{\max \left( \rule{0ex}{2ex} \occ{s [i]}{s [1..i]} - \sigma, 1 \right)}\\
& = \log (n!) - \sum_a \sum_{j = 1}^{\occ{a}{s} - \sigma} \log j\\
& \leq \log (n!) - \sum_a \sum_{j = 1}^{\occ{a}{s}} \log j + \Oh{\sigma^2 \log n}\\
& = \log (n!) - \sum_a \log (\occ{a}{s}!) + \Oh{\sigma^2 \log n}\\
& = \log {n! \choose \occ{a_1}{s}!, \ldots, \occ{a_\sigma}{s}!} + \Oh{\sigma^2 \log n}\\
& \leq n H + \Oh{\sigma^2 \log n}\,,
\end{split}
\end{equation*}
by Lemma~\ref{lem:multinomial}.  If follows that
\begin{equation*}
\begin{split}
& \sum_{i = 1}^n \log \frac{i}{\max \left( \rule{0ex}{2ex} \occ{s [i]}{s [1..i]} - \sigma, 1 \right)} + n + \Oh{\sigma \log n}\\[1ex]
& \leq (H + 1) n + \Oh{\sigma^2 \log n} \qedhere
\end{split}
\end{equation*}
\end{proof}

Our upper bound follows immediately from Corollary~\ref{cor:multinomial}.

\begin{theorem} \label{thm:comp alg}
When \(\sigma = o (n^{1 / 2} / \log n)\), we can sort $s$ online and stably using \((H + 1) n + o (n)\) ternary comparisons and $\Oh{(H + 1) n}$ time.
\end{theorem}

\section{Lower bound} \label{sec:comp lbound}

Consider any online, stable sorting algorithm that uses ternary comparisons.  Since the algorithm is online and stable, it must determine each element's rank relative to the distinct elements already seen, before moving on to the next element.  Since it uses ternary comparisons, we can represent its strategy for each element as an extended binary search tree whose keys are the distinct elements already seen.  If the current element is distinct from all those already seen, then the algorithm reaches a leaf of the tree, and the minimum number of comparisons it performs is equal to that leaf's depth.  If the current element has been seen before, however, then the algorithm stops at an internal node and the minimum number of comparisons it performs is 1 greater than that node's depth; again, the extra comparison is necessary to check that the algorithm should not proceed deeper into the tree.

Suppose \(\sigma = o (n / \log n)\) is a power of 2; then, in any binary search tree on $\sigma$ keys, some key has depth \(\log \sigma \geq H\) (the inequality holds because any distribution on $\sigma$ elements has entropy at most \(\log \sigma\)).  Furthermore, suppose an adversary starts by presenting one copy of each of $\sigma$ distinct element; after that, it considers the algorithm's strategy for the next element as a binary search tree, and presents the deepest key.  This way, the adversary forces the algorithm to use at least \((n - \sigma) (\log \sigma + 1) \geq (H + 1) n - o (n)\) comparisons.

\begin{theorem} \label{thm:comp lbound}
We generally need \((H + 1) n - o (n)\) ternary comparisons to sort $s$ online and stably.
\end{theorem}

\chapter{Online Sorting with Sublinear Memory} \label{chp:sublinear}

When in doubt, sort!  Librarians, secretaries and computer scientists all know that when faced with lots of data, often the best thing is to organize them.  For some applications, though, the data are so overwhelming that we cannot sort.  The streaming model was introduced for situations in which the flow of data cannot be paused or stored in its entirety; the model's assumptions are that we are allowed only one pass over the input and memory sublinear in its size (see, e.g.,~\cite{Mut05}).  Those assumptions mean we cannot sort in general, but in this chapter we show we can when the data are very compressible.

Our inspiration comes from two older articles on sorting. In the first, ``Sorting and searching in multisets'' from 1976, Munro and Spira~\cite{MS76} considered the problem of sorting a multiset $s$ of size $n$ containing $\sigma$ distinct elements in the comparison model. They showed sorting $s$ takes \(\Theta ((H + 1) n)\) time, where \(H = \sum_{i = 1}^\sigma (n_i / n) \log (n / n_i)\) is the entropy of $s$, $\log$ means $\log_2$ and $n_i$ is the frequency of the $i$th smallest distinct element.  When $\sigma$ is small or the distribution of elements in $s$ is very skewed, this is a significant improvement over the \(\Theta (n \log n)\) bound for sorting a set of size $n$.

In the second article, ``Selection and sorting with limited storage'' from 1980, Munro and Paterson~\cite{MP80} considered the problem of sorting a set $s$ of size $n$ using limited memory and few passes.  They showed sorting $s$ in $p$ passes takes \(\Theta (n / p)\) memory locations in the following model (we have changed their variable names for consistency with our own):
\begin{quote}
In our computational model the data is a sequence of $n$ distinct elements stored on a one-way read-only tape.  An element from the tape can be read into one of $r$ locations of random-access storage. The elements are from some totally ordered set (for example the real numbers) and a binary comparison can be made at any time between any two elements within the random-access storage.  Initially the storage is empty and the tape is placed with the reading head at the beginning.  After each pass the tape is rewound to this position with no reading permitted. \dots [I]n view of the limitations imposed by our model, [sorting] must be considered as the \emph{determination} of the sorted order rather than any actual rearrangement.
\end{quote}

An obvious question ---  but one that apparently has still not been addressed decades later --- is how much memory we need to sort a multiset in few passes; in this chapter we consider the case when we are allowed only one pass.  We assume our input is the same as Munro and Spira's, a multiset \(s = \{s_1, \ldots, s_n\}\) with entropy $H$ containing $\sigma$ distinct elements.  To simplify our presentation, we assume \(\sigma \geq 2\) so \(H n = \Omega (\log n)\).  Our model is similar to Munro and Paterson's but it makes no difference to us whether the tape is read-only or read-write, since we are allowed only one pass, and whereas they counted memory locations, we count bits.  We assume machine words are \(\Theta (\log n)\) bits long, an element fits in a constant number of words and we can perform standard operations on words in unit time.  Since entropy is minimized when the distribution is maximally skewed,
\[H n
\geq n \left( \frac{n - \sigma + 1}{n}
    \log \frac{n}{n - \sigma + 1} + \frac{\sigma - 1}{n} \log n \right)
\geq (\sigma - 1) \log n\,;\]
thus, under our assumptions, $\Oh{\sigma}$ words take $\Oh{\sigma \log n} \subseteq \Oh{H n}$ bits.

In Section~\ref{sec:sub alg} we consider the problem of determining the permutation $\pi$ such that \(s_{\pi (1)}, \ldots, s_{\pi (n)}\) is the stable sort of $s$ (i.e., \(s_{\pi (i)} \leq s_{\pi (i + 1)}\) and, if \(s_{\pi (i)} = s_{\pi (i + 1)}\), then \(\pi (i) < \pi (i + 1)\)). For example, if
\[s = a_1, b_1, r_1, a_2, c, a_3, d, a_4, b_2, r_2, a_5\]
(with subscripts serving only to distinguish copies of the same distinct element), then the stable sort of $s$ is
\begin{equation*}
\begin{split}
& a_1, a_2, a_3, a_4, a_5, b_1, b_2, c, d, r_1, r_2\\
& = s_1, s_4, s_6, s_8, s_{11}, s_2, s_9, s_5, s_7, s_3, s_{10}
\end{split}
\end{equation*}
and
\[\pi = 1, 4, 6, 8, 11, 2, 9, 5, 7, 3, 10\,.\]
We give a simple algorithm that computes $\pi$ using one pass, $\Oh{(H + 1) n}$ time and $\Oh{H n}$ bits of memory.  In Section~\ref{sec:sub lbound} we consider the simpler problem of determining a permutation $\rho$ such that \(s_{\rho (1)}, \ldots, s_{\rho (n)}\) is in sorted order (not necessarily stably-sorted).  We prove that in the worst case it takes \(\Omega (H n)\) bits of memory to compute any such $\rho$ in one pass.

\section{Algorithm} \label{sec:sub alg}

The key to our algorithm is the fact \(\pi = \ell_1 \cdots \ell_\sigma\), where $\ell_i$ is the sorted list of positions in which the $i$th smallest distinct element occurs. In our example, \(s = a, b, r, a, c, a, d, a, b, r, a\),
\begin{align*}
\ell_1 & = 1, 4, 6, 8, 11\\
\ell_2 & = 2, 9\\
\ell_3 & = 5\\
\ell_4 & = 7\\
\ell_5 & = 3, 10\,.
\end{align*}

Since each $\ell_i$ is a strictly increasing sequence, we can store it compactly using Elias' gamma code~\cite{Eli75}: we write the first number in $\ell_i$, encoded in the gamma code; for \(1 \leq j < n_i\), we write the difference between the \((j + 1)\)st and $j$th numbers, encoded in the gamma code.  The gamma code is a prefix-free code for the positive integers; for \(x \geq 1\), \(\gamma (x)\) consists of \(\lfloor \log x \rfloor\) zeroes followed by the \((\lfloor \log x \rfloor + 1)\)-bit binary representation of $x$. In our example, we encode $\ell_1$ as
\[\gamma (1)\ \gamma(3)\ \gamma(2)\ \gamma(2)\ \gamma(3)
= 1\ 011\ 010\ 010\ 011\ .\]

\begin{lemma} \label{lem:lists}
We can store $\pi$ in $\Oh{H n}$ bits of memory.
\end{lemma}

\begin{proof}
Encoding the length of every list with the gamma code takes $\Oh{\sigma \log n} \subseteq \Oh{H n}$ bits.  Notice the numbers in each list $\ell_i$ sum to at most $n$.  By Jensen's Inequality, since \(|\gamma (x)| \leq 2 \log x + 1\) and $\log$ is concave, we store $\ell_i$ in at most \(2 n_i \log (n / n_i) + n_i\) bits.  Therefore, storing \(\ell_1, \ldots, \ell_\sigma\) as described above takes
\[\sum_{i = 1}^\sigma (2 n_i \log (n / n_i) + n_i)
= \Oh{(H + 1) n}\]
bits of memory.

To reduce $\Oh{(H + 1) n}$ to $\Oh{H n}$ --- important when one distinct element dominates, so $H$ is close to 0 and \(H n \ll n\) --- we must avoid writing a codeword for each element in $s$.  Notice that, for each run of length at least 2 in $s$ (a run being a maximal subsequence of copies of the same element) there is a run of 1's in the corresponding list $\ell_i$.  We replace each run of 1's in $\ell_i$ by a single 1 and the length of the run.  For each except the last run of each distinct element, the run-length is at most the number we write for the element in $s$ immediately following that run; storing the last run-length for every character takes $\Oh{\sigma \log n} \subseteq \Oh{H n}$ bits.  It follows that storing \(\ell_1, \ldots, \ell_\sigma\) takes
\[\sum_{i = 1}^\sigma \Oh{r_i \log (n / r_i) + r_i}
\leq \sum_{i = 1}^\sigma \Oh{n_i \log (n / n_i)} + \Oh{r}
= \Oh{H n + r}\]
bits, where $r_i$ is the number of runs of the $i$th smallest distinct element and $r$ is the total number of runs in $s$.  M\"{a}kinen and Navarro~\cite{MN05} showed \(r \leq H n + 1\), so our bound is $\Oh{H n}$.
\end{proof}

To compute \(\ell_1, \ldots, \ell_\sigma\) in one pass, we keep track of which distinct elements have occurred and the positions of their most recent occurrences, which takes $\Oh{\sigma}$ words of memory.  For \(1 \leq j \leq n\), if $s_j$ is an occurrence of the $i$th smallest distinct element and that element has not occurred before, then we start $\ell_i$'s encoding with \(\gamma (j)\); if it last occurred in position \(k \leq j - 2\), then we append \(\gamma (j - k)\) to $\ell_i$; if it occurred in position \(j - 1\) but not \(j - 2\), then we append \(\gamma (1)\ \gamma (1)\) to $\ell_i$; if it occurred in both positions \(j - 1\) and \(j - 2\), then we increment the encoded run-length at the end of $\ell_i$'s encoding.  Because we do not know in advance how many bits we will use to encode each $\ell_i$, we keep the encoding in an expandable binary array~\cite{CLRS01}: we start with an array of size 1 bit; whenever the array overflows, we create a new array twice as big, copy the contents from the old array into the new one, and destroy the old array.  We note that appending a bit to the encoding takes amortized constant time.

\begin{lemma} \label{lem:arrays}
We can compute \(\ell_1, \ldots, \ell_\sigma\) in one pass using $\Oh{H n}$ bits of memory.
\end{lemma}

\begin{proof}
Since we are not yet concerned with time, for each element in $s$ we can simply perform a linear search --- which is slow but uses no extra memory --- through the entire list of distinct elements to find the encoding we should extend.  Since an array is never more than twice the size of the encoding it holds, we use $\Oh{H n}$ bits of memory for the arrays.
\end{proof}

To make our algorithm time-efficient, we use search in a splay tree~\cite{ST85} instead of linear search.  At each node of the splay tree, we store a distinct element as the key, the position of that element's most recent occurrence and a pointer to the array for that element.  For \(1 \leq j \leq n\), we search for $s_j$ in the splay tree; if we find it, then we extend the encoding of the corresponding list as described above, set the position of $s_j$'s most recent occurrence to $j$ and splay $s_j$'s node to the root; if not, then we insert a new node storing $s_j$ as its key, position $j$ and a pointer to an expandable array storing \(\gamma (j)\), and splay the node to the root. Figure~\ref{fig:abracadab} shows the state of our splay tree and arrays after we process the first $9$ elements in our example; i.e., \(a, b, r, a, c, a, d, a, b\).  Figure~\ref{fig:abracadabr} shows the changes when we process the next element, an $r$: we double the size of $r$'s array from $4$ to $8$ bits in order to append \(\gamma (10 - 3 = 7) = 00111\), set the position of $r$'s most recent occurrence to $10$ and splay $r$'s node to the root.  Figure~\ref{fig:abracadabra} shows the final state of our splay tree and arrays after we process the last element, an $a$: we append \(\gamma (11 - 8 = 3) = 011\) to $a$'s array (but since only $10$ of its $16$ bits were already used, we do not expand it), set the position of $a$'s most recent occurrence to $11$ and splay $a$'s node to the root.

\begin{figure}
\centering
\rotatebox{270}{\includegraphics[width=5cm]{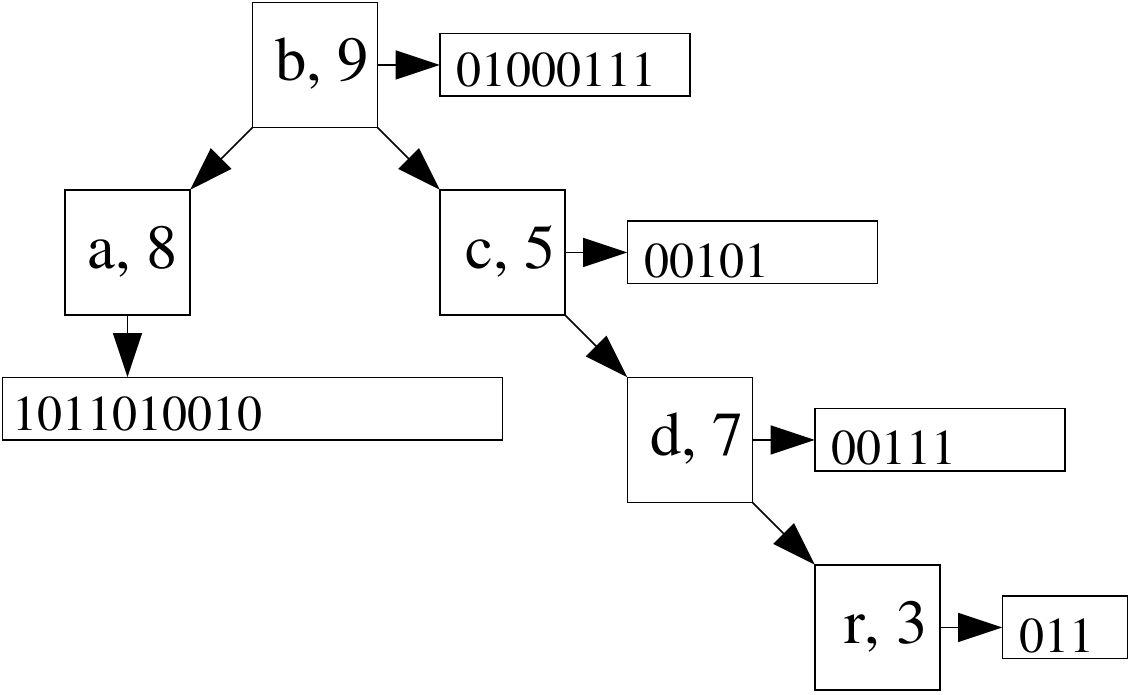}}
\caption{Our splay tree and arrays after we process $a, b, r, a, c, a, d, a, b$.}
\label{fig:abracadab} \vspace{3ex}
\rotatebox{270}{\includegraphics[width=5cm]{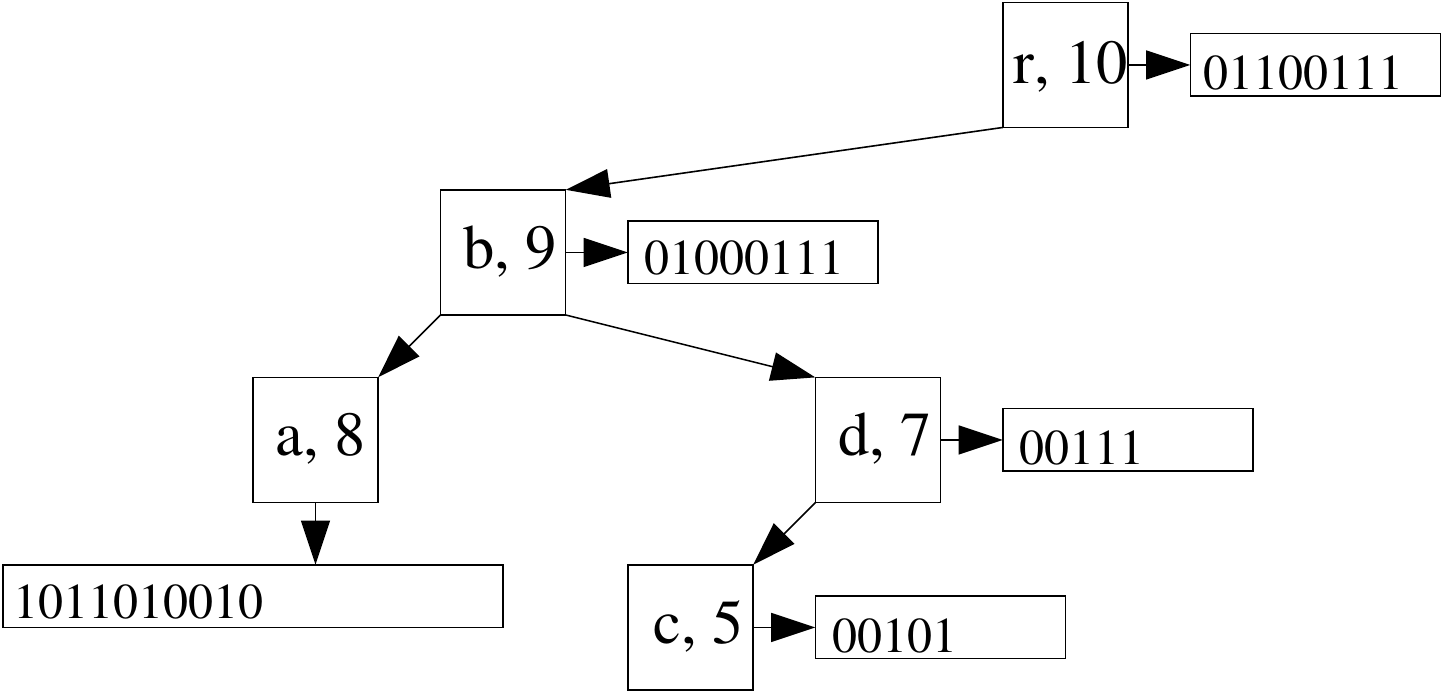}}
\caption{Our splay tree and arrays after we process $a, b, r, a, c, a, d, a, b, r$; notice we have doubled the size of the array for $r$, in order to append $\gamma (7) = 00111$.}
\label{fig:abracadabr} \vspace{3ex}
\rotatebox{270}{\includegraphics[width=6.25cm]{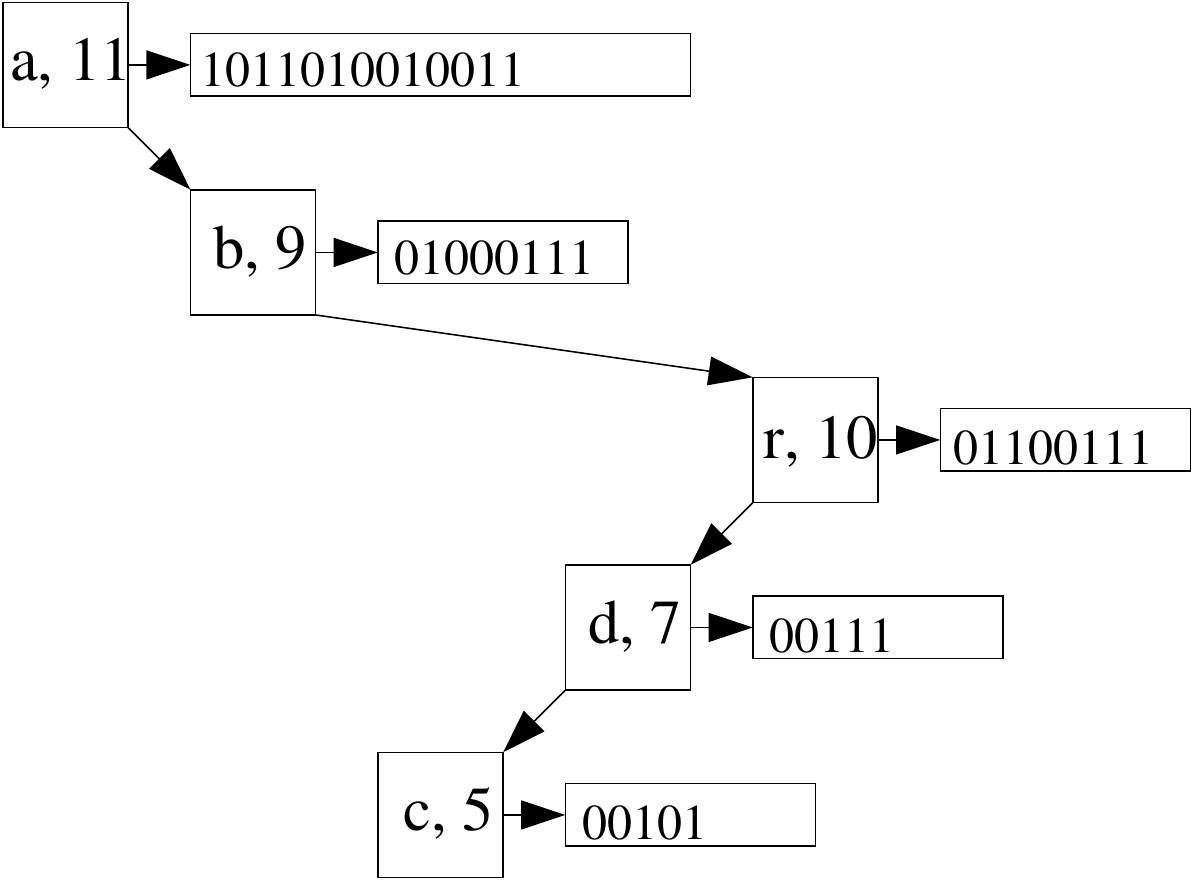}}
\caption{Our splay tree and arrays after we process $a, b, r, a, c, a, d, a, b, r, a$; notice we have not had to expand the array for $a$ in order to append $\gamma (3) = 011$.}
\label{fig:abracadabra}
\end{figure}

\begin{lemma} \label{lem:splay tree}
We can compute \(\ell_1, \ldots, \ell_\sigma\) in one pass using $\Oh{(H + 1) n}$ time and $\Oh{H n}$ bits of memory.
\end{lemma}

\begin{proof}
Our splay tree takes $\Oh{\sigma}$ words of memory and, so, does not change the bound on our memory usage.  For \(1 \leq i \leq \sigma\) we search for the $i$th largest distinct element once when it is not in the splay tree, insert it once, and search for it \(n_i - 1\) times when it is in the splay tree.  Therefore, by the Update Lemma~\cite{ST85} for splay trees, the total time taken for all the operations on the splay tree is
\[\sum_{i = 1}^\sigma \Oh{\log \frac{W}{\min (w_{i - 1}, w_{i + 1})} +
    n_i \log \frac{W}{w_i} + n_i}\,,\]
where \(w_1, \ldots, w_\sigma\) are any positive weights, $W$ is their sum and \(w_0 = w_{\sigma + 1} = \infty\).  Setting \(w_i = n_i\) for \(1 \leq i \leq \sigma\), this bound becomes
\[\sum_{i = 1}^\sigma \Oh{\rule{0ex}{2ex}
    (n_i + 2) (\log (n / n_i) + 1)}
= \Oh{(H + 1) n}\,.\]
Because appending a bit to an array takes amortized constant time, the total time taken for operations on the arrays is proportional to the total length in bits of the encodings, i.e., $\Oh{H n}$.
\end{proof}

After we process all of $s$, we can compute $\pi$ from the state of the splay tree and arrays: we perform an in-order traversal of the splay tree; when we visit a node, we decode the numbers in its array and output their positive partial sums (this takes $\Oh{1}$ words of memory and time proportional to the length in bits of the encoding, because the gamma code is prefix-free); this way, we output the concatenation of the decoded lists in increasing order by element, i.e., \(\ell_1 \cdots \ell_\sigma = \pi\).  In our example, we visit the nodes in the order \(a, b, c, d, r\); when we visit $a$'s node we output
\begin{align*}
\gamma^{-1} (1) & = 1\\
1 + \gamma^{-1} (011) = 1 + 3 & = 4\\
4 + \gamma^{-1} (010) = 4 + 2 & = 6\\
6 + \gamma^{-1} (010) = 6 + 2 & = 8\\
8 + \gamma^{-1} (011) = 8 + 3 & = 11\,.
\end{align*}
Our results in this section culminate in the following theorem:

\begin{theorem} \label{thm:sub alg}
We can compute $\pi$ in one pass using $\Oh{(H + 1) n}$ time and $\Oh{H n}$ bits of memory.
\end{theorem}

We note $s$ can be recovered efficiently from our splay tree and arrays: we start with an empty priority queue $Q$ and insert a copy of each distinct element, with priority equal to the position of its first occurrence (i.e., the first encoded number in its array); for \(1 \leq j \leq n\), we dequeue the element with minimum priority, output it, and reinsert it with priority equal to the position of its next occurrence (i.e., its previous priority plus the next encoded number in its array).  This idea --- that a sorted ordering of $s$ partially encodes it --- is central to our lower bound in the next section.

\section{Lower bound} \label{sec:sub lbound}

Consider any algorithm $A$ that, allowed one pass over $s$, outputs a permutation $\rho$ such that \(s_{\rho (1)}, \ldots, s_{\rho (n)}\) is in sorted order (not necessarily stably-sorted).  Notice $A$ generally cannot output anything until it has read all of $s$, in case $s_n$ is the unique minimum; also, given the frequency of each distinct element, $\rho$ tells us the arrangement of elements in $s$ up to equivalence.

\begin{theorem} \label{thm:sub lbound}
In the worst case, it takes \(\Omega (H n)\) bits of memory to compute any sorted ordering of $s$ in one pass.
\end{theorem}

\begin{proof}
Suppose each \(n_i = n / \sigma\), so \(H = \log \sigma\) and the number of possible distinct arrangements of the elements in $s$ is maximized,
\[\frac{n!}{\prod_{i = 1}^\sigma n_i!}
= \frac{n!}{((n / \sigma)!)^\sigma}\,.\]
It follows that in the worst case $A$ uses at least
\begin{equation*}
\begin{split}
& \log \left( \rule{0ex}{2ex} n! / ((n / \sigma)!)^\sigma \right)\\
& = \log n! - \sigma \log (n / \sigma)!\\
& \geq n \log n - n \log e - \sigma \left( \rule{0ex}{2ex} (n / \sigma) \log (n / \sigma)- (n / \sigma) \log e + \Oh{\log (n / \sigma)} \right)\\
& = n \log \sigma - \Oh{\sigma \log (n / \sigma)}\\
\end{split}
\end{equation*}
bits of memory to store $\rho$; the inequality holds by Stirling's Formula,
\[x \log x - x \log e < \log x! \leq x \log x - x \log e + \Oh{\log x}\,.\]
If \(\sigma = \Oh{1}\) then
\[\sigma \log (n / \sigma) = \Oh{\log n} \subset o (n)\,;\]
otherwise, since \(\sigma \log (n / \sigma)\) is maximized when \(\sigma = n / e\),
\[\sigma \log (n / \sigma) = \Oh{n} \subset o (n \log \sigma)\,;\]
in both cases,
\[n \log \sigma - \Oh{\sigma \log (n / \sigma)}
\geq n \log \sigma - o (n \log \sigma)
\geq \Omega (H n)\,.\qedhere\]
\end{proof}

\chapter{One-Pass Compression} \label{chp:one-pass}

Data compression has come of age in recent years and compression algorithms are now vital in situations unforeseen by their designers.  This has led to a discrepancy between the theory of data compression algorithms and their use in practice: compression algorithms are often designed and analysed assuming the compression and decompression operations can use a ``sufficiently large'' amount of working memory; however, in some situations, particularly in mobile or embedded computing environments, the memory available is very small compared to the amount of data we need to compress or decompress.  Even when compression algorithms are implemented to run on powerful desktop computers, some care is taken to be sure that the compression/decompression of large files do not take over all the RAM of the host machine.  This is usually accomplished by splitting the input in blocks (e.g., {\sf bzip}), using heuristics to determine when to discard the old data (e.g., {\sf compress}, {\sf ppmd}), or by maintaining a ``sliding window'' over the more recently seen data and forgetting the oldest data (e.g., {\sf gzip}).

In this chapter we initiate the theoretical study of space-conscious compression algorithms.  Although data compression algorithms have their own peculiarities, this study belongs to the general field of algorithmics in the streaming model (see, e.g.,~\cite{BBD+02,Mut05}), in which we are allowed only one pass over the input and memory sublinear (possibly polylogarithmic or even constant) in its size.  We prove tight upper and lower bounds on the compression ratio achievable by one-pass algorithms that use an amount of memory independent of the size of the input.  By ``one-pass'', we mean that the algorithms are allowed to read each input symbol only once; hence, if an algorithm needs to access (portions of) the input more than once it must store it---consuming part of its precious working memory.  Our bounds are worst-case and given in terms of the empirical $k$th-order empirical entropy of the input string. More precisely we prove the following results:
\begin{enumerate}
\item[(a)] Let \(\lambda \geq 1\), \(k \geq 0\) and \(\epsilon > 0\) be constants and let $g$ be a function independent of $n$. In the worst case it is impossible to store a string $s$ of length $n$ over an alphabet of size $\sigma$ in \(\lambda H_k (s) n + o (n \log \sigma) + g\) bits using one pass and $\Oh{\sigma^{k + 1 / \lambda - \epsilon}}$ bits of memory.
\item[(b)] Given a \((\lambda H_k (s) + o (n \log \sigma) + g)\)-bit encoding of $s$, it is impossible to recover $s$ using one pass and $\Oh{\sigma^{k + 1 / \lambda - \epsilon}}$ bits of memory.
\item[(c)] Given \(\lambda \geq 1\), \(k \geq 0\) and \(\mu > 0\), we can store $s$ in \(\lambda H_k (s) n + \mu n + \Oh{\sigma^{k + 1 / \lambda} \log \sigma}\) bits using one pass and $\Oh{\sigma^{k + 1 / \lambda} \log^2 \sigma}$ bits of memory, and later recover $s$ using one pass and the same amount of memory.
\end{enumerate}
\noindent While $\sigma$ is often treated as constant in the literature, we treat it as a variable to distinguish between, say, $\Oh{\sigma^{k + 1 / \lambda - \epsilon}}$ and $\Oh{\sigma^{k + 1 / \lambda} \log^2 \sigma}$ bits.  Informally, (a) provides a lower bound to the amount of memory needed to compress a string up to its $k$th-order entropy; (b) tells us the same amount of memory is required also for decompression and implies that the use of a powerful machine for doing the compression does not help if only limited memory is available when decompression takes place; (c) establishes that (a) and (b) are nearly tight.  Notice $\lambda$ plays a dual role: for large $k$, it makes (a) and (b)
inapproximability results --- e.g., we cannot use $\Oh{\sigma^k}$ bits of memory without worsening the compression in terms of \(H_k (s)\) by more than a constant factor; for small $k$, it makes (c) an interesting approximability result --- e.g., we can compress
reasonably well in terms of \(H_0 (s)\) using, say, $\Oh{\sqrt{\sigma}}$ bits of memory. The main difference between the bounds in (a)--(b) and (c) is a $\sigma^\epsilon \log^2 \sigma$ factor in the memory usage.  Since $\mu$ is a constant, \(\mu n \in o (n \log \sigma)\) and the bounds on the encoding's length match.  Note that $\mu$ can be arbitrarily small, but the term $\mu n$ cannot be avoided (Lemma~\ref{lem:nonconvergence}).

We use $s$ to denote the string that we want to compress.  We assume that $s$ has length $n$ and is drawn from an alphabet of size $\sigma$.  Note that we measure memory in terms of alphabet size so $\sigma$ is considered a variable.  The $0$th-order empirical
entropy \(H_0 (s)\) of $s$ is defined as \(H_0 (s) = \sum_a \frac{\occ{a}{s}}{n} \log \frac{n}{\occ{a}{s}}\), where $\occ{a}{s}$ is the number of times character $a$ occurs in $s$; throughout, we write $\log$ to mean $\log_2$ and assume \(0 \log 0 = 0\).  It is well known that $H_0$ is the maximum compression we can achieve using a fixed codeword for each alphabet symbol.  We can achieve a greater compression if the codeword we use for each symbol depends on the $k$ symbols preceding it.  In this case the maximum compression is
bounded by the $k$th-order entropy $H_k (s)$ (see~\cite{KM99} for the formal definition).  We use two properties of $k$th-order entropy in particular:
\begin{itemize}
\item \(H_k (s_1) |s_1| + H_k (s_2) |s_2| \leq H_k (s_1 s_2) |s_1 s_2|\),
\item since \(H_0 (s) \leq \log |\{a\,:\,\mbox{$a$ occurs in $s$}\}|\), we have
\[H_k (s) \leq \log \max_{|w| = k} \{ j\,:\,\mbox{$w$ is followed by $j$ distinct characters in $s$}\}\,.\]
\end{itemize}
We point out that the empirical entropy is defined pointwise for any string and can be used to measure the performance of compression algorithms as a function of the string's structure, thus without any assumption on the input source.  For this reason we say that the bounds given in terms of $H_k$ are worst-case bounds.

Some of our arguments are based on Kolmogorov complexity~\cite{LV08}; the Kolmogorov complexity of $s$, denoted \(K (s)\), is the length in bits of the shortest program that outputs $s$; it is generally incomputable but can be bounded from below by counting arguments (e.g., in a set of $m$ elements, most have Kolmogorov complexity at least \(\log m - \Oh{1}\)). We use two properties of Kolmogorov complexity in particular: if an object can be easily computed from other objects, then its Kolmogorov complexity is at most the sum of theirs plus a constant; and a fixed, finite object has constant Kolmogorov complexity.

\section{Algorithm} \label{sec:op alg}

Move-to-front compression~\cite{BSTW86} is probably the best example of a compression algorithm whose space complexity is independent of the input length: keep a list of the characters that have occurred in decreasing order by recency; store each character in the input by outputting its position in the list (or, if it has not occurred before, its index in the alphabet) encoded in Elias' $\delta$ code, then move it to the front of the list. Move-to-front stores a string $s$ of length $n$ over an alphabet of size $\sigma$ in \(\left( \rule{0ex}{2ex} H_0 (s) + \Oh{\log H_0 (s)} \right) n + \Oh{\sigma \log \sigma}\) bits using one pass and $\Oh{\sigma \log \sigma}$ bits of memory.  When memory is scarce, we can use $\Oh{\sigma^{1 / \lambda} \log \sigma}$ bits of memory by storing only the \(\lceil \sigma^{1 / \lambda} \rceil\) most recent characters; it is easy to see that this increases the number of bits stored by a factor of at most $\lambda$.  On the other hand, note that we can store $s$ in \(\left( \rule{0ex}{2ex} H_k (s) + \Oh{\log H_k (s)} \right) n + \Oh{\sigma^{k + 1} \log \sigma}\) bits by keeping a separate list for each possible context of length $k$; this increases the memory usage by a factor of at most $\sigma^k$. In this section we first use a more complicated algorithm to get a better upper bound: given constants \(\lambda \geq 1\), \(k \geq 0\) and \(\mu> 0\), we can store $s$ in \((\lambda H_k (s) + \mu) n + \Oh{\sigma^{k + 1 / \lambda} \log \sigma}\) bits using one pass and $\Oh{\sigma^{k + 1 / \lambda} \log^2 \sigma}$ bits of memory.

We start with the following lemma --- based on a previous paper~\cite{Gag06b} about compression algorithms' redundancies --- that says we can an approximation $Q$ of a probability distribution $P$ in few bits, so that the relative entropy between $P$ and $Q$ is small.  The relative entropy \(D (P \| Q) = \sum_{i = 1}^\sigma p_i \log (p_i / q_i)\) between \(P = p_1, \ldots, p_\sigma\) and \(Q = q_1, \ldots, q_\sigma\) is the expected redundancy per character of an ideal code for $Q$ when characters are drawn according to $P$.

\begin{lemma} \label{lem:distribution}
Let $s$ be a string of length $n$ over an alphabet of size $\sigma$ and let $P$ be the normalized distribution of characters in $s$. Given $s$ and constants \(\lambda \geq 1\) and \(\mu > 0\), we can store a probability distribution $Q$ with \(D (P \| Q) < (\lambda - 1) H (P) + \mu\) in $\Oh{\sigma^{1 / \lambda} \log (n + \sigma)}$ bits using $\Oh{\sigma^{1 / \lambda} \log (n + \sigma)}$ bits of memory.
\end{lemma}

\begin{proof}
Suppose \(P = p_1, \ldots, p_\sigma\). We can use an $\Oh{n \log n}$-time algorithm due to Misra and Gries~\cite{MG82} (see also~\cite{DLM02,KSP03}) to find the \(t \leq r \sigma^{1 / \lambda}\) values of $i$ such that \(p_i \geq 1 / (r \sigma^{1 / \lambda})\), where \(r = 1 + \frac{1}{2^{\mu / 2} - 1}\), using $\Oh{\sigma^{1 / \lambda} \log \max (n, \sigma)}$ bits of memory; or, since we are not concerned with time in this chapter, we can simply make $\sigma$ passes over $s$ to find these $t$ values.  For each, we store $i$ and \(\lfloor p_i r^2 \sigma \rfloor\); since $r$ depends only on $\mu$, in total this takes $\Oh{\sigma^{1 / \lambda} \log \sigma}$ bits.  This information lets us later recover \(Q
= q_1, \ldots, q_\sigma\) where
\[q_i = \left\{ \begin{array}{ll}
\displaystyle \frac{(1 - 1 / r) \lfloor p_i r^2 \sigma \rfloor}
    {\sum \left\{ \lfloor p_j r^2 \sigma \rfloor\,
    :\,p_j \geq 1 / (r \sigma^{1 / \lambda}) \right\}}
\hspace{3ex} & \mbox{if \(p_i \geq 1 / (r \sigma^{1 / \lambda})\),}\\[3ex]
\displaystyle \frac{1}{r (\sigma - t)} \hspace{3ex} & \mbox{otherwise.}
\end{array} \right.\]
Suppose \(p_i \geq 1 / (r \sigma^{1 / \lambda})\); then \(p_i r^2 \sigma \geq r\). Since \(\sum \left\{ \lfloor p_j r^2 \sigma \rfloor\,:\,p_j \geq 1 / (r \sigma^{1 / \lambda}) \right\} \leq r^2 \sigma\),
\begin{equation*}
\begin{split}
& p_i \log (p_i / q_i)\\
& \leq p_i \log \left( \frac{r}{r - 1} \cdot \frac{p_i r^2 \sigma}
    {\lfloor p_i r^2 \sigma \rfloor} \right)\\
& < 2 p_i \log \frac{r}{r - 1}\\
& = p_i \mu\,.
\end{split}
\end{equation*}
Now suppose \(p_i < 1 / (r \sigma^{1 / \lambda})\); then \(p_i \log (1 / p_i) > (p_i / \lambda) \log \sigma\).  Therefore
\begin{equation*}
\begin{split}
& p_i \log (p_i / q_i)\\
& < p_i \log \left( (\sigma - t) / \sigma^{1 / \lambda} \right)\\
& \leq (\lambda - 1) (p_i / \lambda) \log \sigma\\
& < (\lambda - 1) p_i \log (1 / p_i)\,.
\end{split}
\end{equation*}
Since \(p_i \log (p_i / q_i) < (\lambda - 1) p_i \log (1 / p_i) + p_i \mu\) in both cases, \(D (P \| Q) < (\lambda - 1) H (P) + \mu\).
\end{proof}

Armed with this lemma, we can adapt arithmetic coding to use $\Oh{\sigma^{1 / \lambda} \log (n + \sigma)}$ bits of memory with a specified redundancy per character.

\begin{lemma} \label{lem:0th-order}
Given a string $s$ of length $n$ over an alphabet of size $\sigma$ and constants \(\lambda \geq 1\) and \(\mu > 0\), we can store $s$ in \((\lambda H_0 (s) + \mu) n + \Oh{\sigma^{1 / \lambda} \log (n + \sigma)}\) bits using $\Oh{\sigma^{1 / \lambda} \log (n + \sigma)}$ bits of memory.
\end{lemma}

\begin{proof}
Let $P$ be the normalized distribution of characters in $s$, so \(H (P) = H_0 (s)\).  First, as described in Lemma~\ref{lem:distribution}, we store a probability distribution $Q$ with \(D (P \| Q) < (\lambda - 1) H (P) + \mu / 2\) in $\Oh{\sigma^{1 / \lambda} \log \sigma}$ bits using $\Oh{\sigma^{1 / \lambda} \log (n + \lambda)}$ bits of memory.  Then, we process $s$ in blocks \(s_1, \ldots, s_b\) of length \(\lceil 4 / \mu \rceil\) (except $s_b$ may be shorter).  For \(1 \leq i < b\), we store $s_i$ as the first \(\lceil \log (2 / \Pr [X = s_i]) \rceil\) bits to the right of the binary point in
the binary representation of
\begin{equation*}
\begin{split}
f (s_i) & = \Pr [X < s_i] + \Pr [X = s_i] / 2\\
& = \sum_{j = 1}^{\lceil 4 / \mu \rceil} \Pr \left[ \rule{0ex}{3ex}
    X [1] = s_i [1], \ldots, X [j - 1] = s_i [j - 1], X [j] < s_i [j] \right] +\\
& \quad \Pr [X = s_i] / 2\,,
\end{split}
\end{equation*}
where $X$ is a string of length \(\lceil 4 / \mu \rceil\) chosen randomly according to $Q$, \(X < s_i\) means $X$ is lexicographically less than $s_i$, and \(X [j]\) and \(s_i
[j]\) indicate the indices in the alphabet of the $j$th characters of $X$ and $s_i$, respectively. Notice that, since \(|f (s_i) - f(y)| > \Pr [X = s_i] / 2\) for any string \(y \neq s_i\) of length \(\lceil 4 / \mu \rceil\), these bits uniquely identify \(f (s_i)\) and, thus, $s_i$. Also, since the probabilities in $Q$ are $\Oh{\log \sigma}$-bit numbers, we can compute \(f (s_i)\) from $s_i$ with $\Oh{\sigma}$ additions and $\Oh{1 / \mu} = \Oh{1}$ multiplications using $\Oh{\log \sigma}$ bits of memory. (In fact, with appropriate data structures, $\Oh{\log \sigma}$ additions and $\Oh{1}$ multiplications suffice.) Finally, we store $s_b$ in \(|s_b| \lceil \log \sigma \rceil = \Oh{\log \sigma}\) bits. In total we store $s$ in
\begin{equation*}
\begin{split}
& \sum_{i = 1}^{b - 1} \lceil \log (2 / \Pr [X = s_i]) \rceil +
    \Oh{\sigma^{1 / \lambda} \log \sigma}\\
& \leq \sum_{i = 1}^{b - 1} \left( \sum_{j = 1}^{\lceil 4 / \mu \rceil}
    \log (1 / q_{s_i [j]}) + 2 \right) + \Oh{\sigma^{1 / \lambda} \log \sigma}\\
& = n \sum_{i = 1}^\sigma p_i \log (1 / q_i) + 2 (b - 1) +
    \Oh{\sigma^{1 / \lambda} \log \sigma}\\
& \leq n (D (P \| Q) + H (P)) + \mu n / 2 +
    \Oh{\sigma^{1 / \lambda} \log \sigma}\\
& \leq (\lambda H_0 (s) + \mu) n + \Oh{\sigma^{1 / \lambda} \log \sigma}
\end{split}
\end{equation*}
bits using $\Oh{\sigma^{1 / \lambda} \log (n + \sigma)}$ bits of memory.
\end{proof}

We boost our space-conscious arithmetic coding algorithm to achieve a bound in terms of \(H_k (s)\) instead of \(H_0 (s)\) by running a separate copy for each possible $k$-tuple, just as we boosted move-to-front compression.

\begin{lemma} \label{lem:kth-order}
Given a string $s$ of length $n$ over an alphabet of size $\sigma$ and constants \(\lambda \geq 1\), \(k \geq 0\) and \(\mu > 0\), we can store $s$ in \((\lambda H_k (s) + \mu) n + \Oh{\sigma^{k + 1 / \lambda} \log (n + \sigma)}\) bits using $\Oh{\sigma^{k + 1 / \lambda} \log (n + \sigma)}$ bits of memory.
\end{lemma}

\begin{proof}
We store the first $k$ characters of $s$ in $\Oh{\log \sigma}$ bits then apply Lemma~\ref{lem:0th-order} to subsequences \(s_1, \ldots, s_{\sigma^k}\), where $s_i$ consists of the characters in $s$ that immediately follow occurrences of the lexicographically $i$th possible $k$-tuple.  Notice that although we cannot keep \(s_1, \ldots, s_{\sigma^k}\) in memory, enumerating them as many times as necessary in order to apply Lemma~\ref{lem:0th-order} takes $\Oh{\log \sigma}$ bits of memory.
\end{proof}

To make our algorithm use one pass and to change the \(\log (n + \sigma)\) factor to \(\log \sigma\), we process the input in blocks \(s_1, \ldots, s_b\) of length $\Oh{\sigma^{k + 1 / \lambda} \log \sigma}$.  Notice each individual block $s_i$ fits in memory --- so we can apply Lemma~\ref{lem:kth-order} to it --- and \(\log (|s_i| + \sigma) = \Oh{\log \sigma}\).

\begin{theorem} \label{thm:op ubound}
Given a string $s$ of length $n$ over an alphabet of size $\sigma$ and constants \(\lambda \geq 1\), \(k \geq 0\) and \(\mu > 0\), we can store $s$ in \((\lambda H_k (s) + \mu) n + \Oh{\sigma^{k + 1 / \lambda} \log \sigma}\) bits using one pass and $\Oh{\sigma^{k + 1 / \lambda} \log^2 \sigma}$ bits of memory, and later recover $s$ using one pass and the same amount of memory.
\end{theorem}

\begin{proof}
Let $c$ be a constant such that, by Lemma~\ref{lem:kth-order}, we can store any substring $s_i$ of $s$ in \((\lambda H_k (s_i) + \mu / 2) |s_i| + c \sigma^{k + 1 / \lambda} \log \sigma\) bits using $\Oh{\sigma^{1 / \lambda} \log (|s_i| + \sigma)}$ bits of memory.  We process $s$ in blocks \(s_1, \ldots, s_b\) of length \(\lceil (2 c / \mu) \sigma^{k + 1 / \lambda} \log \sigma \rceil\) (except $s_b$ may be shorter).  Notice each block $s_i$ fits in $\Oh{\sigma^{k + 1 / \lambda} \log^2 \sigma}$ bits of memory.  When we reach $s_i$, we read it into memory, apply Lemma~\ref{lem:kth-order} to it --- using
\[\Oh{\sigma^{k + 1 / \lambda} \log \left( \lceil (2 c / \mu)
    \sigma^{k + 1 / \lambda} \log \sigma \rceil + \sigma \right)}
= \Oh{\sigma^{k + 1 / \lambda} \log \sigma}\]
bits of memory --- then erase it from memory. In total we store $s$ in
\begin{equation*}
\begin{split}
& \sum_{i = 1}^b \left((\lambda H_k (s_i) + \mu / 2) |s_i|
    + c \sigma^{k + 1 / \lambda} \log \sigma \right)\\
& \leq (\lambda H_k (s) + \mu / 2) n + b c \sigma^{k + 1 / \lambda} \log \sigma\\
& \leq (\lambda H_k (s) + \mu) n + c \sigma^{k + 1 / \lambda} \log \sigma
\end{split}
\end{equation*}
bits using $\Oh{\sigma^{k + 1 / \lambda} \log^2 \sigma}$ bits of memory.

Notice the encoding of each block $s_i$ also fits in $\Oh{\sigma^{k + 1 / \lambda} \log^2 \sigma}$ bits of memory. To decode each block later, we read its encoding into memory, search through all possible strings of length \(\lceil (2 c / \mu) \sigma^{k + 1 / \lambda} \log \sigma \rceil\) in lexicographic order until we find the one that yields that encoding --- using $\Oh{\sigma^{k + 1 / \lambda} \log^2 \sigma}$ bits of memory --- and output it.
\end{proof}

The method for decompression in the proof of Theorem~\ref{thm:op ubound} above takes exponential time but is very simple (recall we are not concerned with time here); reversing each step of the compression takes linear time but is slightly more complicated.

\section{Lower bounds} \label{sec:op lbounds}

Theorem~\ref{thm:op ubound} is still weaker than the strongest compression bounds that ignore memory constraints, in two important ways: first, even when \(\lambda = 1\) the bound on the compression ratio does not approach \(H_k (s)\) as $n$ goes to infinity; second, we need to know $k$. It is not hard to prove these weaknesses are unavoidable when using fixed memory.  In this section, we use the idea from these proofs to prove a nearly matching lower bound for compression: in the worst case it is impossible to store a string $s$ of length $n$ over an alphabet of size $\sigma$ in \(\lambda H_k (s) n + o (n \log \sigma) + g\) bits, for any function $g$ independent of $n$, using one encoding pass and $\Oh{\sigma^{k + 1 / \lambda - \epsilon}}$ bits of memory. We close with a symmetric lower bound for decompression.

\begin{lemma} \label{lem:nonconvergence}
Let \(\lambda \geq 1\) be a constant and let $g$ be a function independent of $n$.  In the worst case it is impossible to store a string $s$ of length $n$ in \(\lambda H_0 (s) n + o (n) + g\) bits using one encoding pass and memory independent of $n$.
\end{lemma}

\begin{proof}
Let $A$ be an algorithm that, given $\lambda$, stores $s$ using one pass and memory independent of $n$. Since $A$'s future output depends only on its state and its future input, we can model $A$ with a finite-state machine $M$.  While reading $|M|$ characters of $s$, $M$ must visit some state at least twice; therefore either $M$ outputs at least one bit for every $|M|$ characters in $s$ --- or \(n / |M|\) bits in total --- or for infinitely many strings $M$ outputs nothing.  If $s$ is unary, however, then \(H_0 (s) = 0\).
\end{proof}

\begin{lemma} \label{lem:unknown}
Let $\lambda$ be a constant, let $g$ be a function independent of $n$ and let $b$ be a function independent of $n$ and $k$. In the worst case it is impossible to store a string $s$ of length $n$ over an alphabet of size $\sigma$ in \(\lambda H_k (s) n + o (n \log
\sigma) + g\) bits for all $k\geq 0$ using one pass and $b$ bits of memory.
\end{lemma}

\begin{proof}
Let $A$ be an algorithm that, given $\lambda$, $g$, $b$ and $\sigma$, stores $s$ using $b$ bits of memory. Again, we can model it with a finite-state machine $M$, with \(|M| = 2^b\) and $M$'s Kolmogorov complexity \(K (M) = K (\langle A, \lambda, g, b, \sigma \rangle) + \Oh{1} = \Oh{\log \sigma}\).  (Since $A$, $\lambda$, $g$, and $b$ are all fixed, their Kolmogorov complexities are $\Oh{1}$.)

Suppose $s$ is a periodic string with period \(2 b\) whose repeated substring $r$ has \(K (r) = |r| \log \sigma - \Oh{1}\). We can specify $r$ by specifying $M$, the states $M$ is in when it reaches and leaves any copy of $r$ in $s$, and $M$'s output on that copy of $r$. (If there were another string $r'$ that took $M$ between those states with that output, then we could substitute $r'$ for $r$ in $s$ without changing $M$'s output.)  Therefore $M$ outputs at least
\[K (r) - K (M) - \Oh{\log |M|}
= |r| \log \sigma - \Oh{\log \sigma + b}
= \Omega (|r| \log \sigma)\]
bits for each copy of $r$ in $s$, or \(\Omega (n \log \sigma)\) bits in total. For \(k \geq 2 b\), however, \(H_k (s)\) approaches 0 as $n$ goes to infinity.
\end{proof}

The idea behind these proofs is simple: model a one-pass algorithm with a finite-state machine and evaluate its behaviour on a periodic string.  Nevertheless, combining it with the following simple results --- based on the same previous paper~\cite{Gag06b} as Lemma~\ref{lem:distribution} --- we can easily show a lower bound that nearly matches Theorem~\ref{thm:op ubound}.  (In fact, our proofs are valid even for algorithms that make preliminary passes that produce no output --- perhaps to gather statistics, like Huffman coding~\cite{Huf52} --- followed by a single encoding pass that produces all of the output; once the algorithm begins the encoding pass, we can model it with a finite-state machine.)

\begin{lemma} \label{lem:random}
Let \(\lambda \geq 1\), \(k \geq 0\) and \(\epsilon > 0\) be constants and let $r$ be a randomly chosen string of length \(\lfloor \sigma^{k + 1 / \lambda - \epsilon} \rfloor\) over an alphabet of size $\sigma$. With high probability every possible $k$-tuple is followed by $\Oh{\sigma^{1 / \lambda - \epsilon}}$ distinct characters in $r$.
\end{lemma}

\begin{proof}
Consider a $k$-tuple $w$.  For \(1 \leq i \leq n - k\), let \(X_i = 1\) if the $i$th through \((i + k - 1)\)st characters of $s$ are an occurrence of $w$ and the \((i + k)\)th character in $s$ does not occur in $w$; otherwise \(X_i = 0\). Notice $w$ is followed by at most \(\sum_{i = 1}^{n - k} X_i + k\) distinct characters in $s$ and \(\Pr [X_i = 1\,|\,X_j = 1] \leq 1 / \sigma^k\) and \(\Pr [X_i = 1\,|\,X_j = 0] \leq 1 / (\sigma^k - 1)\) for \(i \neq j\). Therefore, by Chernoff bounds (see~\cite{HR90}) and the union bound, with probability greater than
\[1 - \frac{\sigma^k}
    {2^{6 \lfloor \sigma^{k + 1 / \lambda - \epsilon} \rfloor / (\sigma^k - 1)}}
\geq 1 - \sigma^k / 2^{6 \sigma^{1 / \lambda - \epsilon}}\]
every $k$-tuple is followed by fewer than \(6 \lfloor \sigma^{k + 1 / \lambda - \epsilon} \rfloor / (\sigma^k - 1) + k \leq 12 \sigma^{1 / \lambda - \epsilon} + k\) distinct characters.
\end{proof}

\begin{corollary} \label{cor:random}
Let \(\lambda \geq 1\), \(k \geq 0\) and \(\epsilon > 0\) be constants. There exists a string $r$ of length $\lfloor \sigma^{k + 1 / \lambda - \epsilon} \rfloor$ over an alphabet of size $\sigma$ with \(K (r) = |r| \log \sigma - \Oh{1}\) but \(H_k (r^i) \leq (1 / \lambda - \epsilon) \log \sigma + \Oh{1}\) for \(i \geq 1\).
\end{corollary}

\begin{proof}
If $r$ is randomly chosen, then \(K (r) \geq |r| \log \sigma - 1\) with probability greater than \(1 / 2\) and, by Lemma~\ref{lem:random}, with high probability every possible $k$-tuple is followed by $\Oh{\sigma^{1 / \lambda - \epsilon}}$ distinct characters in $r$; therefore there exists an $r$ with both properties.  Every possible $k$-tuple is followed by at most $k$ more distinct characters in $r^i$ than in $r$ and, thus,
\begin{equation*}
\begin{split}
H_k (r^i) & \leq \log \max_{|w| = k} \left\{ j\,:\,
    \parbox{24ex}{$w$ is followed by $j$
    \newline distinct characters in $r^i$} \right\}\\[2ex]
& \leq \log \Oh{\sigma^{1 / \lambda - \epsilon}}\\
& \leq (1 / \lambda - \epsilon) \log \sigma + \Oh{1}\,. \qedhere
\end{split}
\end{equation*}
\end{proof}

Consider what we get if, for some \(\epsilon > 0\), we allow the algorithm $A$ from Lemma~\ref{lem:unknown} to use $\Oh{\sigma^{k + 1 / \lambda - \epsilon}}$ bits of memory, and evaluate it on the periodic string $r^i$ from Corollary~\ref{cor:random}. Since $r^i$ has period \(\lfloor \sigma^{k + 1 / \lambda - \epsilon} \rfloor\) and its repeated substring $r$ has \(K (r) = |r| \log \sigma - \Oh{1}\), the finite-state machine $M$ outputs at least
\[K (r) - K (M) - \Oh{\log |M|}
= |r| \log \sigma - \Oh{\sigma^{k + 1 / \lambda - \epsilon}}
= |r| \log \sigma - \Oh{|r|}\]
bits for each copy of $r$ in $r^i$, or \(n \log \sigma - \Oh{n}\) bits in total.  Because \(\lambda H_k (r^i) \leq (1 - \epsilon) \log \sigma + \Oh{1}\), this yields the following nearly tight lower bound; notice it matches Theorem~\ref{thm:op ubound} except for a \(\sigma^\epsilon \log^2 \sigma\) factor in the memory usage.

\begin{theorem} \label{thm:op lbound}
Let \(\lambda \geq 1\), \(k \geq 0\) and \(\epsilon > 0\) be constants and let $g$ be a function independent of $n$.  In the worst case it is impossible to store a string $s$ of length $n$ over an alphabet of size $\sigma$ in \(\lambda H_k (s) n + o (n \log \sigma) + g\) bits using one encoding pass and $\Oh{\sigma^{k + 1 / \lambda - \epsilon}}$ bits of memory.
\end{theorem}

\begin{proof}
Let $A$ be an algorithm that, given $\lambda$, $k$, $\epsilon$ and $\sigma$, stores $s$ while using one encoding pass and $\Oh{\sigma^{k + 1 / \lambda - \epsilon}}$ bits of memory; we prove that in the worst case $A$ stores $s$ in more than \((\lambda H_k (s) + \mu) n + o (n \log \sigma) + g\) bits.  Again, we can model it with a finite-state machine $M$, with \(|M| = 2^{\Oh{\sigma^{k + 1 / \lambda - \epsilon}}}\) and \(K (M) = \Oh{\log \sigma}\). Let $r$ be a string of length \(\lfloor \sigma^{k + 1 / \lambda - \epsilon} \rfloor\) with \(K (r) \geq |r| \log \sigma - \Oh{1}\) and \(H_k (r^i) \leq (1 / \lambda - \epsilon) \log \sigma + \Oh{1}\) for \(i \geq 1\), as described in Corollary~\ref{cor:random}, and suppose \(s = r^i\) for some $i$.  We can specify $r$ by
specifying $M$, the states $M$ is in when it reaches and leaves any copy of $r$ in $s$, and $M$'s output on that copy. Therefore $M$ outputs at least
\[K (r) - K (M) - \Oh{\sigma^{k + 1 / \lambda - \epsilon}}
= |r| \log \sigma - \Oh{|r|}\]
bits for each copy of $r$ in $s$, or \(n \log \sigma - \Oh{n}\) bits in total --- which is asymptotically greater than \(\lambda H_k (s) n + o (n \log \sigma) + g \leq (1 - \epsilon) n \log \sigma + o (n \log \sigma) + g\).
\end{proof}

With a good bound on how much memory is needed for compression, we turn our attention to decompression. Good bounds here are equally important, because often data is compressed once by a powerful machine (e.g., a server or base-station) and then transmitted to many weaker machines (clients or agents) who decompress it individually. Fortunately for us, compression and decompression are essentially symmetric. Recall Theorem~\ref{thm:op ubound} says we can recover $s$ from a \(\left( \rule{0ex}{2ex} \lambda H_k (s) + \mu) n + \Oh{\sigma^{k + 1 / \lambda} \log \sigma} \right)\)-bit encoding using one pass and $\Oh{\sigma^{k + 1 / \lambda} \log^2 \sigma}$ bits of memory. Using the same idea about finite-state machines and periodic strings gives us the following nearly matching lower bound:

\begin{theorem} \label{thm:decompression}
Let \(\lambda \geq 1\), \(k \geq 0\) and \(\epsilon > 0\) be constants and let $g$ be a function independent of $n$.  There exists a string $s$ of length $n$ over an alphabet of size $\sigma$ such that, given a \((\lambda H_k (s) n + o (n \log \sigma) + g)\)-bit
encoding of $s$, it is impossible to recover $s$ using one pass and $\Oh{\sigma^{k + 1 / \lambda - \epsilon}}$ bits of memory.
\end{theorem}

\begin{proof}
Let $r$ be a string of length $\lfloor \sigma^{k + 1 / \lambda - \epsilon} \rfloor$ with \(K (r) = |r| \log \sigma - \Oh{1}\) but \(H_k (r^i) \leq (1 / \lambda - \epsilon) \log \sigma + \Oh{1}\) for \(i \geq 1\), as described in Corollary~\ref{cor:random}, and suppose \(s = r^i\) for some $i$.  Let $A$ be an algorithm that, given $\lambda$, $k$, $\epsilon$, $\sigma$ and a \((\lambda H_k (s) n + o (n \log \sigma) + g)\)-bit encoding of $s$, recovers $s$ using one pass; we prove $A$ uses \(\omega (\sigma^{k + 1 / \lambda - \epsilon})\) bits of memory. Again, we can model $A$ with a finite-state machine $M$, with \(\log |M|\) equal to the number of bits of memory $A$ uses and \(K (M) = \Oh{\log \sigma}\).  We can specify $r$ by specifying $M$, the state $M$ is in when it starts outputting any copy of $r$ in $s$, and the bits of the encoding it reads while outputting that copy of $r$; therefore
\begin{equation*}
\begin{split}
K (r) & \leq K (M) + \Oh{\log |M|} + \left( \rule{0ex}{2ex} \lambda H_k (s) n +
    o (n \log \sigma) + g \right) / i\\
& \leq \Oh{\log \sigma} + \Oh{\log |M|} + |r| \left( \rule{0ex}{2ex}
    (1 - \epsilon) \log \sigma + o (\log \sigma) + g / n \right)\\
& \leq (1 - \epsilon) |r| \log \sigma + o (|r| \log \sigma) + \Oh{\log |M|} + g / n\,,
\end{split}
\end{equation*}
so
\[\Oh{\log |M|} + g / n \geq \epsilon |r| \log \sigma - o (|r| \log \sigma)
= \Omega (\sigma^{k + 1 / \lambda - \epsilon} \log \sigma)\ .\]
The theorem follows because $n$ can be arbitrarily large compared to $g$.
\end{proof}

\chapter{Stream Compression} \label{chp:streaming}

Massive datasets seem to expand to fill the space available and, in situations when they no longer fit in memory and must be stored on disk, we may need new models and algorithms.  Grohe and Schweikardt~\cite{GS05} introduced read/write streams to model situations in which we want to process data using mainly sequential accesses to one or more disks.  As the name suggests, this model is like the streaming model (see, e.g.,~\cite{Mut05}) but, as is reasonable with datasets stored on disk, it allows us to make multiple passes over the data, change them and even use multiple streams (i.e., disks).  As Grohe and Schweikardt pointed out, sequential disk accesses are much faster than random accesses --- potentially bypassing the von Neumann bottleneck --- and using several disks in parallel can greatly reduce the amount of memory and the number of accesses needed.  For example, when sorting, we need the product of the memory and accesses to be at least linear when we use one disk~\cite{MP80,GKS07} but only polylogarithmic when we use two~\cite{CY91,GS05}.  Similar bounds have been proven for a number of other problems, such as checking set disjointness or equality; we refer readers to Schweikardt's survey~\cite{Sch07} of upper and lower bounds with one or more read/write streams, Heinrich and Schweikardt's recent paper~\cite{HS08} relating read/write streams to classical complexity theory, and Beame and Hu\`{y}nh-Ng\d{o}c's recent paper~\cite{BH08} on the value of multiple read/write streams for approximating frequency moments.

Since sorting is an important operation in some of the most powerful data compression algorithms, and compression is an important operation for reducing massive datasets to a more manageable size, we wondered whether extra streams could also help us achieve better compression.  In this chapter we consider the problem of compressing a string $s$ of $n$ characters over an alphabet of size $\sigma$ when we are restricted to using \(\log^{\mathcal{O} (1)} n\) bits of memory and \(\log^{\mathcal{O} (1)} n\) passes over the data.  In Section~\ref{sec:universal}, we show how we can achieve universal compression using only one pass over one stream.  Our approach is to break the string into blocks and compress each block separately, similar to what is done in practice to compress large files.  Although this may not usually significantly worsen the compression itself, it may stop us from then building a fast compressed index~\cite{FMMN07} (unless we somehow combine the indexes for the blocks) or clustering by compression~\cite{CV05,FGGMV07} (since concatenating files should not help us compress them better if we then break them into pieces again).  In Section~\ref{sec:grammar-based} we use a vaguely automata-theoretic argument to show one stream is not sufficient for us to achieve good grammar-based compression.  Of course, by ``good'' we mean here something stronger than universal compression: we want the size of our encoding to be at most polynomial in the size of the smallest context-free grammar than generates $s$ and only $s$.  We still do not know whether any constant number of streams is sufficient for us to achieve such compression.  Finally, in Section~\ref{sec:entropy-only} we show that two streams are necessary and sufficient for us to achieve entropy-only bounds.  Along the way, we show we need two streams to find strings' minimum periods or compute the Burrows-Wheeler Transform.  As far as we know, this is the first study of compression with read/write streams, and among the first studies of compression in any streaming model; we hope the techniques we use will prove to be of independent interest.

\section{Universal compression} \label{sec:universal}

An algorithm is called universal with respect to a class of sources if, when a string is drawn from any of those sources, the algorithm's redundancy per character approaches 0 with probability 1 as the length of the string grows.  The class most often considered, and which we consider in this section, is that of stationary, ergodic Markov sources (see, e.g.,~\cite{CT06}).  Since the $k$th-order empirical entropy \(H_k (s)\) of $s$ is the minimum self-information per character of $s$ with respect to a $k$th-order Markov source (see~\cite{Sav97}), an algorithm is universal if it stores any string $s$ in \(n H_k (s) + o (n)\) bits for any fixed $\sigma$ and $k$.  The $k$th-order empirical entropy of $s$ is also our expected uncertainty about a randomly chosen character of $s$ when given the $k$ preceding characters.  Specifically,
\[H_k (s) = \left\{ \begin{array}{ll}
    (1 / n) \sum_a \mathsf{occ} (a, s) \log \frac{n}{\mathsf{occ}(a, s)} &
    \mbox{if \(k = 0\),}\\[1ex]
    (1 / n) \sum_{|w| = k} |w_s| H_0 (w_s) &
    \mbox{otherwise,}
\end{array} \right.\]
where \(\mathsf{occ} (a, s)\) is the number of times character $a$ occurs in $s$, and $w_s$ is the concatenation of those characters immediately following occurrences of $k$-tuple $w$ in $s$.

In a previous paper~\cite{GM07b} we showed how to modify the well-known LZ77 compression algorithm~\cite{ZL77} to use sublinear memory while still storing $s$ in \(n H_k (s) + \Oh{n \log \log n / \log n}\) bits for any fixed $\sigma$ and $k$.  Our algorithm uses nearly linear memory and so does not fit into the model we consider in this chapter, but we mention it here because it fits into some other streaming models (see, e.g.,~\cite{Mut05}) and, as far as we know, was the first compression algorithm to do so.  In the same paper we proved several lower bounds using ideas that eventually led to our lower bounds in Sections~\ref{sec:grammar-based} and~\ref{sec:entropy-only} of this chapter.

\begin{theorem}[Gagie and Manzini, 2007] \label{thm:GM07b}
We can achieve universal compression using one pass over one stream and $\Oh{n / \log^2 n}$ bits of memory.
\end{theorem}

To achieve universal compression with only polylogarithmic memory, we use a recent algorithm due to Gupta, Grossi and Vitter~\cite{GGV08}.  Although they designed it for the RAM model, we can easily turn it into a streaming algorithm by processing $s$ in small blocks and compressing each block separately.

\begin{theorem}[Gupta, Grossi and Vitter, 2008] \label{thm:GGV08}
In the RAM model, we can store any string $s$ in \(n H_k (s) + \Oh{\sigma^k \log n}\) bits, for all $k$ simultaneously, using $\Oh{n}$ time.
\end{theorem}

\begin{corollary} \label{cor:universal}
We can achieve universal compression using one pass over one stream and $\Oh{\log^{1 + \epsilon} n}$ bits of memory.
\end{corollary}

\begin{proof}
We process $s$ in blocks of \(\lceil \log^\epsilon n \rceil\) characters, as follows: we read each block into memory, apply Theorem~\ref{thm:GGV08} to it, output the result, empty the memory, and move on to the next block.  (If $n$ is not given in advance, we increase the block size as we read more characters.)  Since Gupta, Grossi and Vitter's algorithm uses $\Oh{n}$ time in the RAM model, it uses $\Oh{n \log n}$ bits of memory and we use $\Oh{\log^{1 + \epsilon} n}$ bits of memory.  If the blocks are \(s_1, \ldots, s_b\), then we store all of them in a total of
\[\sum_{i = 1}^b \left( |s_i| H_k (s_i) + \Oh{\sigma^k \log \log n} \right)
\leq n H_k (s) + \Oh{\sigma^k n \log \log n / \log^\epsilon n}\]
bits for all $k$ simultaneously.  Therefore, for any fixed $\sigma$ and $k$, we store $s$ in \(n H_k (s) + o (n)\) bits.
\end{proof}

A bound of \(n H_k (s) + \Oh{\sigma^k n \log \log n / \log^\epsilon n}\) bits is not very meaningful when $k$ is not fixed and grows as fast as \(\log \log n\), because the second term is \(\omega (n)\).  Notice, however, that Gupta {\it et al.}'s bound of \(n H_k (s) + \Oh{\sigma^k \log n}\) bits is also not very meaningful when \(k \geq \log n\), for the same reason.  As we will see in Section~\ref{sec:entropy-only}, it is possible for $s$ to be fairly incompressible but still to have \(H_k (s) = 0\) for \(k \geq \log n\).  It follows that, although we can prove bounds that hold for all $k$ simultaneously, those bounds cannot guarantee good compression in terms of \(H_k (s)\) when \(k \geq \log n\).

By using larger blocks --- and, thus, more memory --- we can reduce the $\Oh{\sigma^k n \log \log n / \log^\epsilon n}$ redundancy term in our analysis, allowing $k$ to grow faster than \(\log \log n\) while still having a meaningful bound.  We conjecture that the resulting tradeoff is nearly optimal.  Specifically, using an argument similar to those we use to prove the lower bounds in Sections~\ref{sec:grammar-based} and~\ref{sec:entropy-only}, we believe we can prove that the product of the memory, passes and redundancy must be nearly linear in $n$.  It is not clear to us, however, whether we can modify Corollary~\ref{cor:universal} to take advantage of multiple passes.

\begin{problem} \label{prb:tradeoff}
With multiple passes over one stream, can we achieve better bounds on the memory and redundancy than we can with one pass?
\end{problem}

\section{Grammar-based compression} \label{sec:grammar-based}

Charikar {\it et al.}~\cite{CLL+05} and Rytter~\cite{Ryt03} independently showed how to build a context-free grammar {\sf APPROX} that generates $s$ and only $s$ and is an $\Oh{\log n}$ factor larger than the smallest such grammar {\sf OPT}, which is \(\Omega (\log n)\) bits in size.

\begin{theorem}[Charikar {\it et al.}, 2005; Rytter, 2003] \label{thm:grammar_RAM}
In the RAM model, we can approximate the smallest grammar with \(|\mathsf{APPROX}| = \Oh{|\mathsf{OPT}|^2}\) using $\Oh{n}$ time.
\end{theorem}

\noindent In this section we prove that, if we use only one stream, then in general our approximation must be superpolynomially larger than the smallest grammar.  Our idea is to show that periodic strings whose periods are asymptotically slightly larger than the product of the memory and passes, can be encoded as small grammars but, in general, cannot be compressed well by algorithms that use only one stream.  Our argument is based on the following two lemmas.

\begin{lemma} \label{lem:small grammar}
If $s$ has period $\ell$, then the size of the smallest grammar for that string is $\Oh{\ell + \log n \log \log n}$ bits.
\end{lemma}

\begin{proof}
Let $t$ be the repeated substring and $t'$ be the proper prefix of $t$ such that \(s = t^{\lfloor n / \ell \rfloor} t'\).  We can encode a unary string $X^{\lfloor n / \ell \rfloor}$ as a grammar $G_1$ with $\Oh{\log n}$ productions of total size $\Oh{\log n \log \log n}$ bits.  We can also encode $t$ and $t'$ as grammars $G_2$ and $G_3$ with $\Oh{\ell}$ productions of total size $\Oh{\ell}$ bits.  Suppose $S_1$, $S_2$ and $S_3$ are the start symbols of $G_1$, $G_2$ and $G_3$, respectively.  By combining those grammars and adding the productions \(S_0 \rightarrow S_1 S_3\) and \(X \rightarrow S_2\), we obtain a grammar with $\Oh{\ell + \log n}$ productions of total size $\Oh{\ell + \log n \log \log n}$ bits that maps $S_0$ to $s$.
\end{proof}

\begin{lemma} \label{lem:substring}
Consider a lossless compression algorithm that uses only one stream, and a machine performing that algorithm.  We can compute any substring from
\begin{itemize}
\item its length;
\item for each pass, the machine's memory configurations when it reaches and leaves the part of the stream that initially holds that substring;
\item all the output the machine produces while over that part.
\end{itemize}
\end{lemma}

\begin{proof}
Let $t$ be the substring and assume, for the sake of a contradiction, that there exists another substring $t'$ with the same length that takes the machine between the same configurations while producing the same output.  Then we can substitute $t'$ for $t$ in $s$ without changing the machine's complete output, contrary to our specification that the compression be lossless.
\end{proof}

Lemma~\ref{lem:substring} implies that, for any substring, the size of the output the machine produces while over the part of the stream that initially holds that substring, plus twice the product of the memory and passes (i.e., the number of bits needed to store the memory configurations), must be at least that substring's complexity.  Therefore, if a substring is not compressible by more than a constant factor (as is the case for most strings) and asymptotically larger than the product of the memory and passes, then the size of the output for that substring must be at least proportional to the substring's length.  In other words, the algorithm cannot take full advantage of similarities between substrings to achieve better compression.  In particular, if $s$ is periodic with a period that is asymptotically slightly larger than the product of the memory and passes, and $s$'s repeated substring is not compressible by more than a constant factor, then the algorithm's complete output must be \(\Omega (n)\) bits.  By Lemma~\ref{lem:small grammar}, however, the size of the smallest grammar that generates $s$ and only $s$ is bounded in terms of the period.

\begin{theorem} \label{thm:grammar-based}
With one stream, we cannot approximate the smallest grammar with \(|\mathsf{APPROX}| \leq |\mathsf{OPT}|^{\mathcal{O} (1)}\).
\end{theorem}

\begin{proof}
Suppose an algorithm uses only one stream, $m$ bits of memory and $p$ passes to compress $s$, with \(m p = \log^{\mathcal{O} (1)} n\), and consider a machine performing that algorithm.  Furthermore, suppose $s$ is periodic with period \(\lceil m p \log n \rceil\) and its repeated substring $t$ is not compressible by more than a constant factor.  Lemma~\ref{lem:substring} implies that the machine's output while over a part of the stream that initially holds a copy of $t$, must be \(\Omega (m p \log n - m p) = \Omega (m p \log n)\).  Therefore, the machine's complete output must be \(\Omega (n)\) bits.  By Lemma~\ref{lem:small grammar}, however, the size of the smallest grammar that generates $s$ and only $s$ is \(\Oh{m p \log n + \log n \log \log n} \subset \log^{\mathcal{O} (1)} n\) bits.  Since \(n = \log^{\omega (1)} n\), the algorithm's complete output is superpolynomially larger than the smallest grammar.
\end{proof}

As an aside, we note that a symmetric argument shows that, with only one stream, in general we cannot decode a string encoded as a small grammar.  To prove this, instead of considering a part of the stream that initially holds a copy of the repeated substring $t$, we consider a part that is initially blank and eventually holds a copy of $t$.  We can compute $t$ from the machine's memory configurations when it reaches and leaves that part, so the product of the memory and passes must be greater than or equal to $t$'s complexity.  Also, we note that Theorem~\ref{thm:grammar-based} has the following corollary, which may be of independent interest.

\begin{corollary} \label{cor:period}
With one stream, we cannot find strings' minimum periods.
\end{corollary}

\begin{proof}
Consider the proof of Theorem~\ref{thm:grammar-based}.  If we could find $s$'s minimum period, then we could store $s$ in \(\log^{\mathcal{O} (1)} n\) bits by writing $n$ and one copy of its repeated substring $t$.
\end{proof}

We are currently working on a more detailed argument to show that we cannot even check whether a string has a given period.  Unfortunately, as we noted earlier, our results for this section are still incomplete, as we do not know whether multiple streams are helpful for grammar-based compression.

\begin{problem} \label{prb:grammar streams}
With $\Oh{1}$ streams, can we approximate the smallest grammar well?
\end{problem}

\section{Entropy-only bounds} \label{sec:entropy-only}

Kosaraju and Manzini~\cite{KM99} pointed out that proving an algorithm universal does not necessarily tell us much about how it behaves on low-entropy strings.  In other words, showing that an algorithm encodes $s$ in \(n H_k (s) + o (n)\) bits is not very informative when \(n H_k (s) = o (n)\).  For example, although the well-known LZ78 compression algorithm~\cite{ZL78} is universal, \(|\mathsf{LZ78} (1^n) = \Omega (\sqrt{n})\) while \(n H_0 (1^n) = 0\).  To analyze how algorithms perform on low-entropy strings, we would like to get rid of the \(o (n)\) term and prove bounds that depend only on \(n H_k (s)\).  Unfortunately, this is impossible since, as the example above shows, even \(n H_0 (s)\) can be 0 for arbitrarily long strings.

It is not hard to show that only unary strings have \(H_0 (s) = 0\).  For \(k \geq 1\), recall that \(H_k (s) = (1 / n) \sum_{|w| = k} |w_s| H_0 (w_s)\).  Therefore, \(H_k (s) = 0\) if and only if each distinct $k$-tuple $w$ in $s$ is always followed by the same distinct character.  This is because, if a $w$ is always followed by the same distinct character, then $w_s$ is unary, \(H_0 (w_s) = 0\) and $w$ contributes nothing to the sum in the formula.  Manzini~\cite{Man01} defined the $k$th-order modified empirical entropy \(H_k^* (s)\) such that each context $w$ contributes at least \(\lfloor \log |w_s| \rfloor + 1\) to the sum.  Because modified empirical entropy is more complicated than empirical entropy --- e.g., it allows for variable-length contexts --- we refer readers to Manzini's paper for the full definition.  In our proofs in this chapter, we use only the fact that
\[n H_k (s) \leq n H_k^* (s) \leq n H_k (s) + \Oh{\sigma^k \log n}\,.\]

Manzini showed that, for some algorithms and all $k$ simultaneously, it is possible to bound the encoding's length in terms of only \(n H_k^* (s)\) and a constant that depends only on $\sigma$ and $k$; he called such bounds ``entropy-only''.  In particular, he showed that an algorithm based on the Burrows-Wheeler Transform (BWT)~\cite{BW94} stores any string $s$ in at most
\((5 + \epsilon) n H_k^* (s) + \log n + g_k\) bits for all $k$ simultaneously (since \(n H_k^* (s) \geq \log (n - k)\), we could remove the \(\log n\) term by adding 1 to the coefficient \(5 + \epsilon\)).

\begin{theorem}[Manzini, 2001] \label{thm:Man01}
Using the BWT, move-to-front coding, run-length coding and arithmetic coding, we can achieve an entropy-only bound.
\end{theorem}

\noindent The BWT sorts the characters in a string into the lexicographical order of the suffixes that immediately follow them.  When using the BWT for compression, it is customary to append a special character \$ that is lexicographically less than any in the alphabet.  For a more thorough description of the BWT, we again refer readers to Manzini's paper.  In this section we first show how we can compute and invert the BWT with two streams and, thus, achieve entropy-only bounds.  We then show that we cannot achieve entropy-only bounds with only one stream.  In other words, two streams are necessary and sufficient for us to achieve entropy-only bounds.

One of the most common ways to compute the BWT is by building a suffix array.  In his PhD thesis, Ruhl introduced the StreamSort model~\cite{Ruh03,ADRR04}, which is similar to the read/write streams model with one stream, except that it has an extra primitive that sorts the stream in one pass.  Among other things, he showed how to build a suffix array efficiently in this model.

\begin{theorem}[Ruhl, 2003] \label{thm:Ruh03}
In the StreamSort model, we can build a suffix array using $\Oh{\log n}$ bits of memory and $\Oh{\log n}$ passes.
\end{theorem}

\begin{corollary} \label{cor:BWT+}
With two streams, we can compute the BWT using $\Oh{\log n}$ bits of memory and $\Oh{\log^2 n}$ passes.
\end{corollary}

\begin{proof}
We can compute the BWT in the StreamSort model by appending \$ to $s$, building a suffix array, and replacing each value $i$ in the array by the \((i - 1)\)st character in $s$ (replacing either 0 or 1 by \$, depending on where we start counting).  This takes $\Oh{\log n}$ bits of memory and $\Oh{\log n}$ passes.  Since we can sort with two streams using $\Oh{\log n}$ bits memory and $\Oh{\log n}$ passes (see, e.g.,~\cite{Sch07}), it follows that we can compute the BWT using $\Oh{\log n}$ bits of memory and $\Oh{\log^2 n}$ passes.
\end{proof}

Now suppose we are given a permutation $\pi$ on \(n + 1\) elements as a list \(\pi (1), \ldots,\) \(\pi (n + 1)\), and asked to rank it, i.e., to compute the list \(\pi^0 (1), \ldots, \pi^n (1)\).  This problem is a special case of list ranking (see, e.g.,~\cite{ABD+07}) and has a surprisingly long history.  For example, Knuth~\cite[Solution 24]{Knu98} described an algorithm, which he attributed to Hardy, for ranking a permutation with two tapes.  More recently, Bird and Mu~\cite{BM04} showed how to invert the BWT by ranking a permutation.  Therefore, reinterpreting Hardy's result in terms of the read/write streams model gives us the following bounds.

\begin{theorem}[Hardy, c.~1967] \label{thm:Har67}
With two streams, we can rank a permutation using $\Oh{\log n}$ bits of memory and $\Oh{\log^2 n}$ passes.
\end{theorem}

\begin{corollary} \label{cor:BWT-}
With two streams, we can invert the BWT using $\Oh{\log n}$ bits of memory and $\Oh{\log^2 n}$ passes.
\end{corollary}

\begin{proof}
The BWT has the property that, if a character is the $i$th in \(\mathsf{BWT} (s)\), then its successor in $s$ is the lexicographically $i$th in \(\mathsf{BWT} (s)\) (breaking ties by order of appearance).  Therefore, we can invert the BWT by replacing each character by its lexicographic rank, ranking the resulting permutation, replacing each value $i$ by the $i$th character of \(\mathsf{BWT} (s)\), and rotating the string until \$ is at the end.  This takes $\Oh{\log n}$ memory and $\Oh{\log^2 n}$ passes.
\end{proof}

Since we can compute and invert move-to-front, run-length and arithmetic coding using $\Oh{\log n}$ bits of memory and $\Oh{1}$ passes over one stream, by combining Theorem~\ref{thm:Man01} and Corollaries~\ref{cor:BWT+} and~\ref{cor:BWT-} we obtain the following theorem.

\begin{theorem} \label{thm:eo ubound}
With two streams, we can achieve an entropy-only bound using $\Oh{\log n}$ bits of memory and $\Oh{\log^2 n}$ passes.
\end{theorem}

To show we need at least two streams to achieve entropy-only bounds, we use De Bruijn cycles in a proof similar to the one for Theorem~\ref{thm:grammar-based}.  A $k$th-order De Bruijn cycle~\cite{DeB46} is a cyclic sequence in which every possible $k$-tuple appears exactly once. For example, Figure~\ref{fig:cycles} shows a 3rd-order and a 4th-order De Bruijn cycle.  (We need consider only binary De Bruijn cycles.)  Our argument this time is based on Lemma~\ref{lem:substring} and the following results about De Bruijn cycles.

\begin{figure}[h]
\begin{center}
\begin{tabular}{c@{\hspace{20ex}}c}
\begin{tabular}{cccc}
& 0 & 0 &\\
1 &&& 0\\
0 &&& 1\\
& 1 & 1 &
\end{tabular}
& \begin{tabular}{cccccc}
1 & 0 & 0 & 0 & 0 & 1\\
1 &&&&& 0\\
1 &&&&& 0\\
1 & 0 & 1 & 0 & 1 & 1
\end{tabular}
\end{tabular}
\caption{Examples of 3rd-order and 4th-order De Bruijn cycles.}
\label{fig:cycles}
\end{center}
\end{figure}

\begin{lemma} \label{lem:entropy}
If \(s \in d^*\) for some $k$th-order De Bruijn cycle $d$, then \(n H_k^* (s) = \Oh{2^k \log n}\).
\end{lemma}

\begin{proof}
By definition, each distinct $k$-tuple is always followed by the same distinct character; therefore, \(n H_k (s) = 0\) and \(n H_k^* (s) = \Oh{2^k \log n}\).
\end{proof}

\begin{theorem}[De Bruijn, 1946] \label{thm:DeB46}
There are \(2^{2^{k - 1} - k}\) $k$th-order De Bruijn cycles.
\end{theorem}

\begin{corollary} \label{cor:cycle lbound}
We cannot store most $k$th-order De Bruijn cycles in \(o (2^k)\) bits.
\end{corollary}

Since there are $2^k$ possible $k$-tuples, $k$th-order De Bruijn cycles have length $2^k$, so Corollary~\ref{cor:cycle lbound} means that we cannot compress most De Bruijn cycles by more than a constant factor.  Therefore, we can prove a lower bound similar to Theorem~\ref{thm:grammar-based} by supposing that $s$'s repeated substring is a De Bruijn cycle, then using Lemma~\ref{lem:entropy} instead of Lemma~\ref{lem:small grammar}.

\begin{theorem} \label{thm:eo lbound}
With one stream, we cannot achieve an entropy-only bound.
\end{theorem}

\begin{proof}
As in the proof of Theorem~\ref{thm:grammar-based}, suppose an algorithm uses only one stream, $m$ bits of memory and $p$ passes to compress $s$, with \(m p = \log^{\mathcal{O} (1)} n\), and consider a machine performing that algorithm.  This time, however, suppose $s$ is periodic with period $2^{\lceil \log (m p \log n) \rceil}$ and that its repeated substring $t$ is a $k$th-order De Bruijn cycle, \(k = \lceil \log (m p \log n) \rceil\), that is not compressible by more than a constant factor. Lemma~\ref{lem:substring} implies that the machine's output while over a part of the stream that initially holds a copy of $t$, must be \(\Omega (m p \log n - m p) = \Omega (m p \log n)\).  Therefore, the machine's complete output must be \(\Omega (n)\) bits.  By Lemma~\ref{lem:entropy}, however, \(n H_k^* (s) = \Oh{2^k \log n} = \Oh{m p \log^2 n} \subset \log^{\mathcal{O} (1)} n\).
\end{proof}

Notice Theorem~\ref{thm:eo lbound} implies a lower bound for computing the BWT: if we could compute the BWT with one stream then, since we can compute move-to-front, run-length and arithmetic coding using $\Oh{\log n}$ bits of memory and $\Oh{1}$ passes over one stream, we could thus achieve an entropy-only bound with one stream, contradicting Theorem~\ref{thm:eo lbound}.

\begin{corollary} \label{cor:BWT_lbound}
With one stream, we cannot compute the BWT.
\end{corollary}

Grohe and Schweikardt~\cite{GS05} proved that, with $\Oh{1}$ streams, we generally cannot sort \(n / \log n\) numbers, each consisting of \(\log n\) bits, using $\Oh{n^{1 - \epsilon}}$ bits of memory and \(o (\log n)\) passes.  Combining this result with the following lemma, we immediately obtain a lower bound for computing the BWT with $\Oh{1}$ streams.

\begin{lemma} \label{lem:many streams}
With two or more streams, sorting $\Oh{n / \log n}$ numbers, each of \(\log n\) bits, takes $\Oh{\log n}$ more bits of memory and $\Oh{1}$ more passes than computing the BWT of a ternary string of length $n$.
\end{lemma}

\begin{proof}
We reduce the problem of sorting a sequence \(x_1, \ldots, x_m\) of \((\log n)\)-bit binary numbers, \(m = n / (2 \log n + \log \log n + 2)\), to the problem of computing the BWT of a ternary string of length $n$.  Let \(x_i [j]\) denote the $j$th bit of $x_i$.  Using two streams, \(O (1)\) passes and \(O (\log n)\) memory, we replace each \(x_i [j]\) by \(x_i [j]\ 2\ x_i\ i\ j\), writing 2 as a single character, $x_i$ and $i$ each as \(\log n\) bits, and $j$ as \(\log \log n\) bits.  Let $X$ be the resulting string and consider the last \(m \log n\) characters of the BWT of $X$: they are a permutation of the characters followed by 2s in $X$, i.e., the bits of \(x_1, \ldots, x_m\); if \(x_i < x_{i'}\) or \(x_i = x_{i'}\) but \(i < i'\) then, because \(2\ x_i\ i\) is lexicographically less than \(2\ x_{i'}\ i'\), each bit of $x_i$ comes before each bit of $x_{i'}$; if \(j < j'\) then, for any $i$, because \(2\ x_i\ i\ j\) is lexicographically less than \(2\ x_i\ i\ j'\), the bit \(x_i [j]\) comes before the bit \(x_i [j']\).  In other words, the last \(m \log n\) characters of the BWT of $X$ are \(x_1, \ldots, x_m\) in sorted order.
\end{proof}

\begin{corollary} \label{cor:many streams}
With $\Oh{1}$ streams, we cannot compute the BWT of a ternary string of length $n$ using $\Oh{n^{1 - \epsilon}}$ bits of memory and \(o (\log n)\) passes.
\end{corollary}

In another paper~\cite{GM07a} we improved the coefficient in Manzini's bound from \(5 + \epsilon\) to 2.7, using a variant of distance coding instead of move-to-front and run-length coding.  We conjecture this algorithm can also be implemented with two streams.

\begin{problem} \label{prb:coefficient}
With $\Oh{1}$ streams, can we achieve the same entropy-only bounds that we achieve in the RAM model?
\end{problem}

The main idea of distance coding~\cite{Bin00} is to write the starting position of each maximal run (i.e., subsequence consisting of copies of the same character), by writing the distance from the start of each maximal run to the start of the next maximal run of the same character.  Notice we do not need to write the length of each run because the end of each run (except the last) is the position before the start of the next one.  By symmetry, it makes essentially no difference to the length of the encoding if we write the distance to the start of each maximal run from the start of the previous maximal run of the same character, which is not difficult with $\Oh{\sigma \log n}$ bits of memory and $\Oh{1}$ passes.

Kaplan, Landau and Verbin~\cite{KLV07} showed how, using the BWT followed by distance coding and arithmetic coding, we can store $s$ in \(1.73 n H_k (s) + \Oh{\log n}\) bits for any fixed $\sigma$ and $k$.  This bound holds only when we use an idealized arithmetic coder with $\Oh{\log n}$ total redundancy; if we use a 0th-order coder with per character redundancy $\mu$, then the bound becomes \(1.73 n H_k (s) + \mu n + \Oh{\log n}\).  In our paper we used a lemma due to M\"akinen and Navarro~\cite{MN05} bounding the number of runs in terms of the of the product of the length and the 0th-order empirical entropy, to change the latter bound into \((1.73 + \mu) n H_k (s) + \Oh{\log n}\), which is an improvement when \(H_k (s) < 1\).  Unfortunately, the presence of the $\Oh{\log n}$ term prevents this from being an entropy-only bound.  To prove an entropy-only bound, we modified distance coding to use an escape mechanism, which we have not verified can be implemented in the read/write streams model.

\chapter{Conclusions and Future Work} \label{chp:conclusions}

In this thesis we have tried to provide a fairly complete but coherent view of our studies of sequential-access data compression, balancing discussion of previous work with presentation of our own results.  We would like to highlight now what we consider our key ideas.  The most important innovation in Chapter~\ref{chp:adaptive} was probably our use of predecessor queries for encoding and decoding with a canonical code.  This, combined with our use of Fredman and Willard's data structure~\cite{FW93}, Shannon coding and background processing, allowed us to encode and decode each character in constant worst-case time while producing an encoding whose length was worst-case optimal.  Chapters~\ref{chp:comparisons} and~\ref{chp:sublinear} were, admittedly, somewhat tangential to our topic, but we included them to show how our interests shifted from the model we considered in Chapter~\ref{chp:adaptive} to the one we considered in Chapter~\ref{chp:one-pass}.  The key idea in Chapter~\ref{chp:one-pass} was to view one-pass algorithms with memory bounded in terms of the alphabet size and context length as finite-state machines.  This, combined with the fact that short, randomly chosen strings almost certainly have low empirical entropy, allowed us to prove a lower bound on the amount of memory needed to achieve good compression, that nearly matched our upper bound (which was relatively easy to prove, given Lemma~\ref{lem:distribution}).  Finally, the key idea in Chapter~\ref{chp:streaming} was to to extend the automata-theoretic arguments of Chapter~\ref{chp:one-pass} to algorithms that can make multiple passes and use an amount of memory that depends on the length of the input.  This gave us our lower bound for achieving good grammar-based compression with one stream, our lower bound for finding strings' minimum periods and, combined with properties of De Bruijn sequences, our lower bound for achieving entropy-only bounds.

As we mentioned in the introduction, a paper~\cite{GKN09} we wrote with Marek Karpinski and Yakov Nekrich at the University of Bonn that partially combining the results in Chapters~\ref{chp:adaptive} and~\ref{chp:one-pass}, will appear at the 2009 Data Compression Conference.  This paper concerns fast adaptive prefix coding with memory bounded in terms of the alphabet size and context length, and shows that we can encode $s$ in \((\lambda H + \Oh{1}) n + o (n)\) bits while using $\Oh{\sigma^{1 / \lambda + \epsilon}}$ bits of memory and $\Oh{\log \log \sigma}$ worst-case time to encode and decode each character.  Of course, we would like to improve these bounds, and perhaps implement and test how our algorithm performs with large alphabets such as Chinese, Unicode or the English vocabulary.  We would also like to implement our algorithm from Chapter~\ref{chp:adaptive}, testing several implementations of dictionaries to determine which is the fastest in practice; Fredman and Willard's analysis has enormous constants hidden in the asymptotic notation.  Finally, we are preparing a paper with Nekrich that will give efficient algorithms for adaptive alphabetic prefix coding, adaptive prefix coding for unequal letter costs, and adaptive length-restricted prefix coding (see~\cite{Gag07a} for descriptions of these problems).

As we also mentioned in the introduction, we are currently investigating whether we can prove any more results like the lower bound in Chapter~\ref{chp:streaming} on finding strings' minimum periods.  We are working on the open problems presented in Chapter~\ref{chp:streaming}, about using multiple passes to obtain smaller redundancy terms for universal compression with one stream, approximating the smallest grammar with $\Oh{1}$ streams, and achieving better entropy-only bounds with $\Oh{1}$ streams.  Finally, we have been collaborating with Paolo Ferragina at the University of Pisa and Giovanni Manzini at the University of Eastern Piedmont on a paper~\cite{FGM} about BWT-based compression in the external memory model (see~\cite{Vit08}) with limited random disk accesses.  For the moment, however, our curiosity about sequential-access data compression is mostly satisfied.

After proving our first results about adaptive prefix coding~\cite{Gag07a}, we wrote several papers~\cite{Gag06a,Gag06b,Gag08,Gaga} concerning the number of bits needed to store a good approximation of a probability distribution and, more generally, a Markov process.  For one of these papers~\cite{Gag06b}, about bounds on the redundancy in terms of the alphabet size and context length, we proved versions of Lemmas~\ref{lem:distribution} and~\ref{lem:random}, which eventually led to Chapters~\ref{chp:one-pass} and~\ref{chp:streaming}.  We are now curious whether our results can be combined with algorithms that build sophisticated probabilistic models, either for data compression (see, e.g.,~\cite{FGMS05,FGM06,FM08}) or for inference (see, e.g.,~\cite{Ris86,RST96,AB00,BY01} and subsequent articles).  These algorithms work by considering a class of probabilistic models that are, essentially, Markov sources with variable-length contexts, and finding the model that minimizes the sum of the length of the model's description and the self-information of the input with respect to the model; we note this sum is something like the $k$th-order modified empirical entropy.  Similar kinds of models are used in both applications because many algorithms for inference are based on Rissanen's Minimum Description Length Principle~\cite{Ris78}, which is based on ideas from data compression.

How we minimize the sum of the length of the model's description and the self-information depends on how we represent the model.  At least some of the algorithms mentioned above assume that the length of description is proportional to the number of contexts used.  However, it seems that, if some contexts occur frequently but the distributions of characters that follow them are nearly uniform, and other occur rarely but are always followed by the same character, then it might give better compression to prune the former and keep the latter, which take only $\Oh{\log \sigma}$ bits each to store.  Of course, this is just speculation at the moment.

\bibliographystyle{alpha}
\cleardoublepage
\addcontentsline{toc}{chapter}{Bibliography}
\bibliography{thesis}

\newcommand{\etalchar}[1]{$^{#1}$}
\begin{thebibliography}{FMMN07}

\bibitem[AB00]{AB00}
A.~Apostolico and G.~Bejerano.
\newblock Optimal amnesic probabilistic automata, or, {How} to learn and
  classify proteins in linear time and space.
\newblock {\em Computational Biology}, 7(3--4):381--393, 2000.

\bibitem[ABD{\etalchar{+}}07]{ABD+07}
L.~Arge, M.~A. Bender, E.~D. Demaine, B.~{Holland-Minkley}, and J.~I. Munro.
\newblock An optimal cache-oblivious priority queue and its application to
  graph algorithms.
\newblock {\em SIAM Journal on Computing}, 36(6):1672--1695, 2007.

\bibitem[ABT99]{ABT99}
A.~Andersson, P.~{Bro Miltersen}, and M.~Thorup.
\newblock Fusion trees can be implemented with {AC$^0$} instructions only.
\newblock {\em Theoretical Computer Science}, 215(1--2):337--344, 1999.

\bibitem[ADRR04]{ADRR04}
G.~Aggarwal, M.~Datar, S.~Rajagopalan, and M.~Ruhl.
\newblock On the streaming model augmented with a sorting primitive.
\newblock In {\em Proceedings of the 45th Symposium on Foundations of Computer
  Science}, pages 540--549, 2004.

\bibitem[AL62]{AL62}
G.~{Adelson-Velskii} and E.~M. Landis.
\newblock An algorithm for the organization of information.
\newblock {\em Doklady Akademi Nauk}, 146:263--266, 1962.

\bibitem[BBD{\etalchar{+}}02]{BBD+02}
B.~Babcock, S.~Babu, M.~Datar, R.~Motwani, and J.~Widom.
\newblock Models and issues in data stream systems.
\newblock In {\em Proceedings of the 21st Symposium on Principles of Database
  Systems}, pages 1--16, 2002.

\bibitem[BH08]{BH08}
P.~Beame and D.-T. {Hu\`{y}nh-Ng\d{o}c}.
\newblock On the value of multiple read/write streams for approximating
  frequency moments.
\newblock In {\em Proceedings of the 49th Symposium on Foundations of Computer
  Science}, pages 499--508, 2008.

\bibitem[Bin00]{Bin00}
E.~Binder.
\newblock Distance coder.
\newblock Usenet group {\sf comp.compression}, 2000.

\bibitem[BM04]{BM04}
R.~S. Bird and {S.-C.} Mu.
\newblock Inverting the {Burrows-Wheeler} transform.
\newblock {\em Journal of Functional Programming}, 14(6):603--612, 2004.

\bibitem[BSTW86]{BSTW86}
J.~L. Bentley, D.~D. Sleator, R.~E. Tarjan, and V.~K. Wei.
\newblock A locally adaptive data compression scheme.
\newblock {\em Communications of the ACM}, 29:320--330, 1986.

\bibitem[BW94]{BW94}
M.~Burrows and D.~J. Wheeler.
\newblock A block-sorting lossless data compression algorithm.
\newblock Technical Report~24, Digital Equipment Corporation, 1994.

\bibitem[BY01]{BY01}
G.~Bejerano and G.~Yona.
\newblock Variations on probabilistic suffix trees: {Statistical} modeling and
  prediction of protein families.
\newblock {\em Bioinformatics}, 17(1):23--43, 2001.

\bibitem[CLL{\etalchar{+}}05]{CLL+05}
M.~Charikar, E.~Lehman, D.~Liu, R.~Panigrahy, M.~Prabhakaran, A.~Sahai, and
  A.~Shelat.
\newblock The smallest grammar problem.
\newblock {\em IEEE Transactions on Information Theory}, 51(7):2554--2576,
  2005.

\bibitem[CLRS01]{CLRS01}
T.~H. Cormen, C.~E. Leiserson, R.~L. Rivest, and C.~Stein.
\newblock {\em Introduction to Algorithms}.
\newblock MIT Press, 2nd edition, 2001.

\bibitem[CT06]{CT06}
T.~M. Cover and J.~A. Thomas.
\newblock {\em Elements of Information Theory}.
\newblock Wiley, 2nd edition, 2006.

\bibitem[CV05]{CV05}
R.~Cilibrasi and P.~Vit\'{a}nyi.
\newblock Clustering by compression.
\newblock {\em IEEE Transactions on Information Theory}, 51(4):1523--1545,
  2005.

\bibitem[CY91]{CY91}
J.~Chen and C.-K. Yap.
\newblock Reversal complexity.
\newblock {\em SIAM Journal on Computing}, 20(4):622--638, 1991.

\bibitem[dB46]{DeB46}
N.~G. de~Bruijn.
\newblock A combinatorial problem.
\newblock {\em Koninklijke Nederlandse Akademie van Wetenschappen},
  49:758--764, 1946.

\bibitem[DLM02]{DLM02}
E.~D. Demaine, A.~{L{\'o}pez-Ortiz}, and J.~I. Munro.
\newblock Frequency estimation of {Internet} packet streams with limited space.
\newblock In {\em Proceedings of the 10th European Symposium on Algorithms},
  pages 384--360, 2002.

\bibitem[DM80]{DM80}
D.~P. Dobkin and J.~I. Munro.
\newblock Determining the mode.
\newblock {\em Theoretical Computer Science}, 12:255--263, 1980.

\bibitem[Eli75]{Eli75}
P.~Elias.
\newblock Universal codeword sets and representations of the integers.
\newblock {\em IEEE Transactions on Information Theory}, 21(2):194--203, 1975.

\bibitem[EM{\c{S}}04]{EMS04}
F.~Erg{\"u}n, S.~Muthukrishnan, and S.~C. {\c{S}}ahinalp.
\newblock Sublinear methods for detecting periodic trends in data streams.
\newblock In {\em Proceedings of the 6th Latin American Symposium on
  Theoretical Informatics}, pages 16--28, 2004.

\bibitem[Fal73]{Fal73}
N.~Faller.
\newblock An adaptive system for data compression.
\newblock In {\em Record of the 7th Asilomar Conference on Circuits, Systems
  and Computers}, pages 593--597, 1973.

\bibitem[FGG{\etalchar{+}}07]{FGGMV07}
P.~Ferragina, R.~Giancarlo, V.~Greco, G.~Manzini, and G.~Valiente.
\newblock Compression-based classification of biological sequences and
  structures via the {Universal} {Similarity} {Metric}: {Experimental}
  assessment.
\newblock {\em BMC Bioinformatics}, 8:252, 2007.

\bibitem[FGM]{FGM}
P.~Ferragina, T.~Gagie, and G.~Manzini.
\newblock Space-conscious data indexing and compression in a streaming model.
\newblock In preparation.

\bibitem[FGM06]{FGM06}
P.~Ferragina, R.~Giancarlo, and G.~Manzini.
\newblock The engineering of a compression boosting library: {Theory} vs.\
  practice in {BWT} compression.
\newblock In {\em Proceedings of the 14th European Symposium on Algorithms},
  pages 756--767, 2006.

\bibitem[FGMS05]{FGMS05}
P.~Ferragina, R.~Giancarlo, G.~Manzini, and M.~Sciortino.
\newblock Boosting textual compression in optimal linear time.
\newblock {\em Journal of the ACM}, 52(4):688--713, 2005.

\bibitem[Fis84]{Fis84}
T.~M. Fischer.
\newblock On entropy decomposition methods and algorithm design.
\newblock {\em Colloquia Mathematica Societatis J{\'{a}}nos Bolyai},
  44:113--127, 1984.

\bibitem[FM08]{FM08}
P.~Ferragina and G.~Manzini.
\newblock Boosting textual compression.
\newblock In M.-Y. Kao, editor, {\em Encyclopedia of Algorithms}. Springer,
  2008.

\bibitem[FMMN07]{FMMN07}
P.~Ferragina, G.~Manzini, V.~M{\"a}kinen, and G.~Navarro.
\newblock Compressed representations of sequences and full-text indexes.
\newblock {\em ACM Transactions on Algorithms}, 3(2), 2007.

\bibitem[FW93]{FW93}
M.~L. Fredman and D.~E. Willard.
\newblock Surpassing the information theoretic bound with fusion trees.
\newblock {\em Journal of Computer and System Sciences}, 47(3):424--436, 1993.

\bibitem[Gaga]{Gaga}
T.~Gagie.
\newblock Compressed depth sequences.
\newblock {\em Theoretical Computer Science}.
\newblock To appear.

\bibitem[Gagb]{Gagb}
T.~Gagie.
\newblock On the value of multiple read/write streams for data compression.
\newblock Submitted.

\bibitem[Gag04]{Gag04}
T.~Gagie.
\newblock Dynamic {Shannon} coding.
\newblock In {\em Proceedings of the 12th European Symposium on Algorithms},
  pages 359--370, 2004.

\bibitem[Gag06a]{Gag06a}
T.~Gagie.
\newblock Compressing probability distributions.
\newblock {\em Information Processing Letters}, 97(4):246--251, 2006.

\bibitem[Gag06b]{Gag06b}
T.~Gagie.
\newblock Large alphabets and incompressibility.
\newblock {\em Information Processing Letters}, 99:246--251, 2006.

\bibitem[Gag07a]{Gag07a}
T.~Gagie.
\newblock Dynamic {Shannon} coding.
\newblock {\em Information Processing Letters}, 102(2--3):113--117, 2007.

\bibitem[Gag07b]{Gag07b}
T.~Gagie.
\newblock Sorting streamed multisets.
\newblock In {\em Proceedings of the 10th Italian Conference on Theoretical
  Computer Science}, pages 130--138, 2007.

\bibitem[Gag08]{Gag08}
T.~Gagie.
\newblock Sorting streamed multisets.
\newblock {\em Information Processing Letters}, 108(6):418--421, 2008.

\bibitem[Gal78]{Gal78}
R.~G. Gallager.
\newblock Variations on a theme by {Huffman}.
\newblock {\em IEEE Transactions on Information Theory}, 24(6):668--674, 1978.

\bibitem[GGV08]{GGV08}
A.~Gupta, R.~Grossi, and J.~S. Vitter.
\newblock Nearly tight bounds on the encoding length of the {Burrows-Wheeler
  Transform}.
\newblock In {\em Proceedings of the 4th Workshop on Analytic Algorithmics and
  Combinatorics}, pages 191--202, 2008.

\bibitem[GKN09]{GKN09}
T.~Gagie, M.~Karpinski, and Y.~Nekrich.
\newblock In {\em Proceedings of the Data Compression Conference}, 2009.
\newblock To appear.

\bibitem[GKS07]{GKS07}
M.~Grohe, C.~Koch, and N.~Schweikardt.
\newblock Tight lower bounds for query processing on streaming and external
  memory data.
\newblock {\em Theoretical Computer Science}, 380(1--3):199--217, 2007.

\bibitem[GM07a]{GM07a}
T.~Gagie and G.~Manzini.
\newblock Move-to-front, distance coding, and inversion frequencies revisited.
\newblock In {\em Proceedings of the 18th Symposium on Combinatorial Pattern
  Matching}, pages 71--82, 2007.

\bibitem[GM07b]{GM07b}
T.~Gagie and G.~Manzini.
\newblock Space-conscious compression.
\newblock In {\em Proceedings of the 32nd Symposium on Mathematical Foundations
  of Computer Science}, pages 206--217, 2007.

\bibitem[GN]{GN}
T.~Gagie and Y.~Nekrich.
\newblock Worst-case optimal adaptive prefix coding.
\newblock arXiv:0812.3306.

\bibitem[GS05]{GS05}
M.~Grohe and N.~Schweikardt.
\newblock Lower bounds for sorting with few random accesses to external memory.
\newblock In {\em Proceedings of the 24th Symposium on Principles of Database
  Systems}, pages 238--249, 2005.

\bibitem[GSU]{GSU}
R.~Giancarlo, D.~Scaturro, and F.~Utro.
\newblock Textual data compression in the {-omic} sciences: A sysopsis.
\newblock Submitted.

\bibitem[HR90]{HR90}
T.~Hagerup and C.~R{\"u}b.
\newblock A guided tour of {Chernoff} bounds.
\newblock {\em Information Processing Letters}, 33(6):305--308, 1990.

\bibitem[HS08]{HS08}
A.~Hernich and N.~Schweikardt.
\newblock Reversal complexity revisited.
\newblock {\em Theoretical Computer Science}, 401(1--3):191--205, 2008.

\bibitem[Huf52]{Huf52}
D.~A. Huffman.
\newblock A method for the construction of minimum-redundancy codes.
\newblock {\em Proceedings of the IRE}, 40(9):1098--1101, 1952.

\bibitem[HV92]{HV92}
P.~G. Howard and J.~S. Vitter.
\newblock Analysis of arithmetic coding for data compression.
\newblock {\em Information Processing and Management}, 28(6):749--764, 1992.

\bibitem[Kle00]{Kle00}
S.~T. Klein.
\newblock Skeleton trees for the efficient decoding of {Huffman} encoded texts.
\newblock {\em Information Retrieval}, 3(1):7--23, 2000.

\bibitem[KLV07]{KLV07}
H.~Kaplan, S.~Landau, and E.~Verbin.
\newblock A simpler analysis of {Burrows-Wheeler}-based compression.
\newblock {\em Theoretical Computer Science}, 387(3):220--235, 2007.

\bibitem[KM99]{KM99}
R.~Kosaraju and G.~Manzini.
\newblock Compression of low entropy strings with {Lempel-Ziv} algorithms.
\newblock {\em SIAM Journal on Computing}, 29(3):893--911, 1999.

\bibitem[KN]{KN}
M.~Karpinski and Y.~Nekrich.
\newblock A fast algorithm for adaptive prefix coding.
\newblock {\em Algorithmica}.
\newblock To appear.

\bibitem[Knu85]{Knu85}
D.~E. Knuth.
\newblock Dynamic {Huffman} coding.
\newblock {\em Journal of Algorithms}, 6(2):163--180, 1985.

\bibitem[Knu98]{Knu98}
D.~E. Knuth.
\newblock {\em The Art of Computer Programming}, volume~3.
\newblock Addison-Wesley, 2nd edition, 1998.

\bibitem[KSP03]{KSP03}
R.~M. Karp, S.~Shenker, and C.~H. Papadimitriou.
\newblock A simple algorithm for finding frequent elements in streams and bags.
\newblock {\em ACM Transactions on Database Systems}, 28(1):51--55, 2003.

\bibitem[LV08]{LV08}
M.~Li and P.~Vit{\'a}nyi.
\newblock {\em An Introduction to Kolmogorov Complexity and Its Applications}.
\newblock Springer-Verlag, 3rd edition, 2008.

\bibitem[Man01]{Man01}
G.~Manzini.
\newblock An analysis of the {Burrows-Wheeler Transform}.
\newblock {\em Journal of the ACM}, 48(3):407--430, 2001.

\bibitem[Meh77]{Meh77}
K.~Mehlhorn.
\newblock A best possible bound for the weighted path length of binary search
  trees.
\newblock {\em SIAM Journal on Computing}, 6(2):235--239, 1977.

\bibitem[MG82]{MG82}
J.~Misra and D.~Gries.
\newblock Finding repeated elements.
\newblock {\em Science of Computer Programming}, 2(2):143--152, 1982.

\bibitem[MLP99]{MLP99}
R.~L. Milidi{\'u}, E.~S. Laber, and A.~A. Pessoa.
\newblock Bounding the compression loss of the {FGK} algorithm.
\newblock {\em Journal of Algorithms}, 32(2):195--211, 1999.

\bibitem[MN05]{MN05}
V.~M{\"{a}}kinen and G.~Navarro.
\newblock Succinct suffix arrays based on run-length encoding.
\newblock {\em Nordic Journal of Computing}, 12(1):40--66, 2005.

\bibitem[MP80]{MP80}
J.~I. Munro and M.~S. Paterson.
\newblock Selection and sorting with limited storage.
\newblock {\em Theoretical Computer Science}, 12:315--323, 1980.

\bibitem[MR91]{MR91}
J.~I. Munro and V.~Raman.
\newblock Sorting multisets and vectors in-place.
\newblock In {\em Proceedings of the 2nd Workshop on Algorithms and Data
  Structures}, pages 473--480, 1991.

\bibitem[MS76]{MS76}
J.~I. Munro and P.~M. Spira.
\newblock Sorting and searching in multisets.
\newblock {\em SIAM Journal on Computing}, 5(1):1--8, 1976.

\bibitem[Mut05]{Mut05}
S.~Muthukrishnan.
\newblock {\em Data Streams: Algorithms and Applications}.
\newblock Foundations and Trends in Theoretical Computer Science. Now
  Publishers, 2005.

\bibitem[Nek07]{Nek07}
Y.~Nekrich.
\newblock An efficient implementation of adaptive prefix coding.
\newblock In {\em Proceedings of the Data Compression Conference}, page 396,
  2007.

\bibitem[Ris78]{Ris78}
J.~Rissanen.
\newblock Modeling by shortest data description.
\newblock {\em Automatica}, 14:465--471, 1978.

\bibitem[Ris86]{Ris86}
J.~Rissanen.
\newblock Complexity of strings in the class of {Markov} sources.
\newblock {\em IEEE Transactions on Information Theory}, 32(4):526--532, 1986.

\bibitem[RO04]{RO04}
L.~G. Rueda and B.~J. Oommen.
\newblock A nearly-optimal {Fano}-based coding algorithm.
\newblock {\em Information Processing and Management}, 40(2):257--268, 2004.

\bibitem[RO06]{RO06}
L.~Rueda and B.~J. Oommen.
\newblock A fast and efficient nearly-optimal adaptive {Fano} coding scheme.
\newblock {\em Information Sciences}, 176(12):1656--1683, 2006.

\bibitem[RO08]{RO08}
L.~Rueda and B.~J. Oommen.
\newblock An efficient compression scheme for data communication which uses a
  new family of self-organizing binary search trees.
\newblock {\em International Journal of Communication Systems},
  21(10):1091--1120, 2008.

\bibitem[Rob55]{Rob55}
H.~Robbins.
\newblock A remark on {Stirling's Formula}.
\newblock {\em American Mathematical Monthly}, 62(1):26--29, 1955.

\bibitem[RST96]{RST96}
D.~Ron, Y.~Singer, and N.~Tishby.
\newblock The power of amnesia: {Learning} probabilistic automata with variable
  memory length.
\newblock {\em Machine Learning}, 25(2--3):117--149, 1996.

\bibitem[Ruh03]{Ruh03}
J.~M. Ruhl.
\newblock {\em Efficient Algorithms for New Computational Models}.
\newblock PhD thesis, Massachusetts Institute of Technology, 2003.

\bibitem[Ryt03]{Ryt03}
W.~Rytter.
\newblock Application of {Lempel-Ziv} factorization to the approximation of
  grammar-based compression.
\newblock {\em Theoretical Computer Science}, 302(1--3):211--222, 2003.

\bibitem[Sav97]{Sav97}
S.~Savari.
\newblock Redundancy of the {Lempel-Ziv} incremental parsing rule.
\newblock {\em IEEE Transactions on Information Theory}, 43(1):9--21, 1997.

\bibitem[Sch07]{Sch07}
N.~Schweikardt.
\newblock Machine models and lower bounds for query processing.
\newblock In {\em Proceedings of the 26th Symposium on Principles of Database
  Systems}, pages 41--52, 2007.

\bibitem[Sha48]{Sha48}
C.~E. Shannon.
\newblock A mathematical theory of communication.
\newblock {\em Bell System Technical Journal}, 27:379--423, 623--656, 1948.

\bibitem[SK64]{SK64}
E.~S. Schwartz and B.~Kallick.
\newblock {\em Communications of the ACM}, 7(3):166--169, 1964.

\bibitem[ST85]{ST85}
D.~D. Sleator and R.~E. Tarjan.
\newblock Self-adjusting binary search trees.
\newblock {\em Journal of the ACM}, 32(3):652--686, 1985.

\bibitem[Tho03]{Tho03}
M.~Thorup.
\newblock On {AC$^0$} implementations of fusion trees and atomic heaps.
\newblock In {\em Proceedings of the 14th Symposium on Discrete Algorithms},
  pages 699--707, 2003.

\bibitem[TM00]{TM00}
A.~Turpin and A.~Moffat.
\newblock Housekeeping for prefix coding.
\newblock {\em IEEE Transactions on Communications}, 48(4), 2000.

\bibitem[TM01]{TM01}
A.~Turpin and A.~Moffat.
\newblock On-line adaptive canonical prefix coding with bounded compression
  loss.
\newblock {\em IEEE Transactions on Information Theory}, 47(1):88--98, 2001.

\bibitem[Vit87]{Vit87}
J.~S. Vitter.
\newblock Design and analysis of dynamic {Huffman} codes.
\newblock {\em Journal of the ACM}, 1987(4):825--845, 1987.

\bibitem[Vit08]{Vit08}
J.~S. Vitter.
\newblock {\em Algorithms and Data Structures for External Memory}.
\newblock Foundations and Trends in Theoretical Computer Science. Now
  Publishers, 2008.

\bibitem[ZL77]{ZL77}
J.~Ziv and A.~Lempel.
\newblock A universal algorithm for sequential data compression.
\newblock {\em IEEE Transactions on Information Theory}, 23(3):337--343, 1977.

\bibitem[ZL78]{ZL78}
J.~Ziv and A.~Lempel.
\newblock Compression of individual sequences via variable-rate coding.
\newblock {\em IEEE Transactions on Information Theory}, 24(5):530--536, 1978.

\end{thebibliography}

\end{document}